\documentclass[submission,copyright,creativecommons]{eptcs}
\providecommand{\event}{WPTE 2016} 
\usepackage{latexsym}
\usepackage{amssymb}
\usepackage{color}
\usepackage{theorem}
\usepackage[all]{xy}
\usepackage{amsmath}

\DeclareMathAlphabet{\mathcal}{OMS}{cmsy}{m}{n}

\newcommand{\eop}{\hspace*{\fill}$\Box$}
\def\qed{\eop}
\newtheorem{definition}{Definition}[section]
\newtheorem{lemma}[definition]{Lemma}
\newtheorem{Copiedlemma}{Lemma}
\newtheorem{Copiedtheorem}{Theorem}

\newtheorem{theorem}[definition]{Theorem}
{\theorembodyfont{\rmfamily}
  \newtheorem{example}[definition]{\it Example}
  \newtheorem{myproof}{\it Proof.}
  \newtheorem{myproofsketch}{\it Proof (Sketch).}
}

\newenvironment{proof}{\begin{myproof}}{\end{myproof}}
\newenvironment{proofsketch}{\begin{myproofsketch}}{\end{myproofsketch}}

\def\Serbanuta{{\c S}erb{\u a}nu{\c t}{\u a}}
\def\Rosu{Ro{\c s}u}
\def\SRTransformation{SR Transformation}
\def\SRtransformation{SR transformation}
\def\cC{\mathcal{C}}
\def\cD{\mathcal{D}}
\def\cF{\mathcal{F}}

\def\cR{\mathcal{R}}
\def\cV{\mathcal{V}}
\def\Root{\mathit{root}}
\def\Arity{\mathit{arity}}
\def\Var{{\mathcal{V}\mathit{ar}}}

\def\Dom{{\mathcal{D}\mathit{om}}}
\def\Range{{\mathcal{R}\mathit{an}}}
\def\Subst{{\mathcal{S}\mathit{ub}}}
\def\Condition#1#2{#1 \tto #2}
\def\MyVec#1#2#3{\overrightarrow{#1_{#2..#3}}}
\def\MapVec#1#2#3#4{\overrightarrow{#1(#2_{#3..#4})}}
\def\Numv#1#2{|#1|_{#2}}
\def\ID{\mathtt{id}}
\def\Extension{\mathtt{ext}}
\def\ol#1{\overline{#1}}
\def\wh#1{\widehat{#1}}
\def\Reset#1{(#1)^\bot}
\def\U{\mathbb{U}}
\def\Usymb{U}
\def\SR{\mathbb{SR}}

\def\cRaux{\cR_{\mathit{aux}}}
\def\symb#1{\mathsf{#1}}
\def\olsymb#1{\ol{\symb{#1}}}
\def\guard#1{\langle #1 \rangle}
\def\T{\mathbb{T}}
\def\Tuple{\symb{tuple}}

\def\Convert{\Phi}
\def\Auxconvert{\Psi}
\newcommand{\phito}[1][\U(\cR)]{\Rightarrow_{\Convert,#1}}
\def\tto{\twoheadrightarrow}
\def\cRqsort{\cR_1}
\def\cRsub{\cR_2}
\def\cRqsortaux{\cR_3}
\def\cRidea{\cR_4}
\def\Append{\mathop{+\!\!\!\!+}}
\def\UsplitA{\symb{u}_1}
\def\EvalsplitA#1{[#1]_1}
\def\UsplitB{\symb{u}_2}
\def\EvalsplitB#1{[#1]_2}
\def\UsplitC{\symb{u}_3}
\def\EvalsplitC#1{[#1]_3}
\def\UsplitD{\symb{u}_4}
\def\EvalsplitD#1{[#1]_4}
\def\Uqsort{\symb{u}_5}
\def\Evalqsort#1{[#1]_5}
\def\UfC{\symb{u}_6}

\def\UfD{\symb{u}_7}

\def\UfC{\symb{u}_6}

\newcommand{\arSR}[1][d]{\ar[#1]^(.6){\SR(\cRidea)}}
\newcommand{\arU}[1][d]{\ar[#1]^(.6){\U(\cRidea)}}
\newcommand{\arUl}[1][d]{\ar[#1]_(.6){\U(\cRidea)}}
\newcommand{\arUs}[1][d]{\ar[#1]_{+\hspace{-.5ex}}^(.7){\U(\cRidea)}}
\newcommand{\Enc}[2][]{\save [].[#1]!C="Enc#2"*+[F.:<3pt>]\frm{} \restore}

\title{Sound Structure-Preserving Transformation 
for Weakly-Left-Linear Deterministic Conditional Term Rewriting Systems}
\def\titlerunning{Sound Structure-Preserving Transformation for Weakly-Left-Linear DCTRSs}
\author{
Ryota Nakayama
\qquad\qquad
Naoki Nishida
\qquad\qquad
Masahiko Sakai
\institute{Graduate School of Information Science\\
Nagoya University\\
Nagoya, Japan}
\email{\mbox{\{}nakayama@trs.cm.\mbox{, }nishida@\mbox{, }sakai@\mbox{\}}is.nagoya-u.ac.jp}
}
\def\authorrunning{R. Nakayama, N. Nishida, and M. Sakai}
\begin{document}
\maketitle

\begin{abstract}
In this paper, we show that the {\SRtransformation}, a computationally
equivalent transformation proposed by {\Serbanuta} and {\Rosu}, is a
sound structure-preserving transformation for weakly-left-linear
deterministic conditional term rewriting systems. 
More precisely, we show that every weakly-left-linear deterministic
conditional term rewriting system can be converted to an equivalent
weakly-left-linear and ultra-weakly-left-linear deterministic
conditional term rewriting system and prove that the {\SRtransformation}
is sound for weakly-left-linear and ultra-weakly-left-linear
deterministic conditional term rewriting systems. 
Here, soundness for a conditional term rewriting system means that
reduction of the transformed unconditional term rewriting system
creates no undesired reduction sequence for the conditional system.
\end{abstract}

\section{Introduction}
\label{sec:intro}

\emph{Conditional term rewriting} is known to be much more complicated
than unconditional term rewriting in the sense of analyzing properties,
e.g., \emph{operational termination}~\cite{LMM05},
\emph{confluence}~\cite{SMI95}, and \emph{reachability}~\cite{FG03}.
A popular approach to the analysis of conditional term rewriting systems
(CTRS) is to transform a CTRS into an unconditional term rewriting
system (TRS) that is in general an overapproximation of the CTRS in
terms of reduction. 
Such an approach enables us to use techniques for the analysis of TRSs,
which have been well investigated in the literature.
For example, if the transformed TRS is terminating, then the CTRS is
operationally terminating~\cite{DLMMU04}---to prove termination of the
transformed TRS, we can use many termination proving techniques that
have been well investigated for TRSs (cf.~\cite{Ohl02}).

There are two approaches to transformations of CTRSs into TRSs:
\emph{unravelings}~\cite{Mar96,Mar97} proposed by Marchiori (see,
e.g.,~\cite{GG08,NSS12lmcs}), and a transformation~\cite{Vir99} proposed
by Viry (see, e.g.,~\cite{SR06,GG08}).

Unravelings are transformations from a CTRS into a TRS over an extension
of the original signature for the CTRS, which are \emph{complete} for
(reduction of) the CTRS~\cite{Mar96}.
Here, completeness for a CTRS means that for every reduction sequence of
the CTRS, there exists a corresponding reduction sequence of the
unraveled TRS.
In this respect, the unraveled TRS is an overapproximation of the CTRS
w.r.t.\ reduction, and is useful for analyzing the properties of the
CTRS, such as syntactic properties, modularity, and operational
termination, since TRSs are in general much easier to handle than CTRSs.

The latest transformation based on Viry's approach is a
\emph{computationally equivalent} transformation proposed by
{\Serbanuta} and {\Rosu}~\cite{SR06,SR06b} (the {\SRtransformation}, for
short), which is one of \emph{structure-preserving}
transformations~\cite{GN14wpte}.
This transformation has been proposed for normal CTRSs
in~\cite{SR06}---started with this class to simplify the
discussion---and then been extended to \emph{strongly or syntactically
deterministic} CTRSs (SDCTRSs) that are \emph{ultra-left-linear}
(\emph{semilinear}~\cite{SR06b}).
Here, for a syntactic property \textit{P}, a CTRS is said to be
\emph{ultra-\textit{P}} if its unraveled TRS via Ohlebusch's
unraveling~\cite{Ohl01} has the property \textit{P}.
The {\SRtransformation} converts 
a confluent, operationally terminating, and ultra-left-linear SDCTRS
into a TRS that is \emph{computationally equivalent} to the CTRS. 
This means that such a converted TRS can be used to exactly simulate any
reduction sequence of the original CTRS to a normal form.

As for unravelings, soundness of the {\SRtransformation} plays a very
important role for, e.g., computational equivalence. 
Here, soundness for a CTRS means that reduction of the converted TRS
creates no undesired reduction sequences for the CTRS.
Neither any unraveling nor the {\SRtransformation} is sound for all
CTRSs.
Since soundness is one of the most important properties for
transformations of CTRSs, sufficient conditions for soundness have been
well investigated, especially for unravelings (see,
e.g.,~\cite{GGS10,NSS12lmcs,GGS12}).
For example, the \emph{simultaneous unraveling} that has been proposed
by Marchiori~\cite{Mar96} (and then has been improved by
Ohlebusch~\cite{Ohl01}) is sound for \emph{weakly-left-linear} (WLL, for
short), \emph{confluent}, \emph{non-erasing}, or \emph{ground
conditional} normal CTRSs~\cite{GGS10}, and for DCTRSs that are
\emph{confluent and right-stable}, \emph{WLL}, or
\emph{ultra-right-linear}~\cite{GGS12}.
Normal CTRSs admit a rewrite rule to have conditions to test terms
received via variables in the left-hand side, e.g., whether a term with
such variables can reach a ground normal form specified by the rule. 
This means that we can add so-called \emph{guard} conditions to rewrite rules.
In addition to such a function, DCTRSs admit a rewrite rule to have
so-called \texttt{let}-structures in functional languages. 
On the other hand, the WLL property allows CTRSs to have rules, e.g.,
$\symb{eq}(x,x) \to \symb{true}$, to test equivalence between terms via
non-linear variables. 
For these reasons, the class of WLL DCTRSs is one of the most
interesting and practical classes of CTRSs, as well as that of WLL
normal CTRSs.

The main purpose of transformations along the Viry's approach is to use
the soundly transformed TRS in order to simulate the reduction of the
original CTRS.
The experimental results in~\cite{SR06} indicate that the rewriting
engine using the soundly transformed TRS is much more efficient than the
one using the original left-linear normal CTRS. 
To get an efficient rewriting engine for CTRSs, soundness conditions for
the {\SRtransformation} are worth investigating.

In the case of DCTRSs that are not normal CTRSs, the {\SRtransformation}
is defined for \emph{ultra-left-linear} SDCTRSs, and has been shown to
be sound for such SDCTRSs~\cite{SR06b}.
On the other hand, unlike unravelings, soundness conditions for the
{\SRtransformation} have been investigated \emph{only} for normal
CTRSs~\cite{SR06,SR06b,NYG14wpte}.
For example, it has been shown in~\cite{NYG14wpte} that the
{\SRtransformation} is sound for WLL normal CTRSs, but the result has
not been adapted to WLL SDCTRSs yet. 

In this paper, we show that the {\SRtransformation} is a sound
structure-preserving transformation for WLL DCTRSs that do not have to
be SDCTRSs.
To this end, we first show that every WLL DCTRSs can be converted to a
WLL and ultra-WLL DCTRS such that the reductions of these DCTRSs are the
same.
Then, we show that the {\SRtransformation} is applicable to ultra-WLL
DCTRSs without any change.
Finally, we prove that the {\SRtransformation} is sound for WLL and
ultra-WLL DCTRSs.
These results imply that the composition of the conversion to ultra-WLL
DCTRSs and the {\SRtransformation} is a sound structure-preserving
transformation for WLL DCTRSs.

The contribution of this paper is summarized as follows.
We adapt the result on soundness of the {\SRtransformation} for WLL
normal CTRSs to WLL \emph{deterministic} CTRSs.
The result in this paper covers the result in~\cite{NYG14wpte} for WLL
normal CTRSs showing a simpler proof that would be helpful for further
development of the {\SRtransformation} and its soundness.

This paper is organized as follows. 
In Section~\ref{sec:preliminaries}, we briefly recall basic notions and
notations of term rewriting.
In Section~\ref{sec:transformations}, we recall the notion of soundness,
the simultaneous unraveling, and the {\SRtransformation} for DCTRSs, and
show that every WLL DCTRS can be converted to an equivalent WLL and
ultra-WLL DCTRS.
In Section~\ref{sec:relative-soundness}, we show that the
{\SRtransformation} is sound for WLL and ultra-WLL DCTRSs. 
In Section~\ref{sec:conclusion}, we conclude this paper and describe
future work on this research. 
Some missing proofs are available at
\url{http://www.trs.cm.is.nagoya-u.ac.jp/~nishida/wpte16/}.

\section{Preliminaries}
\label{sec:preliminaries}

In this section, we recall basic notions and notations of term
rewriting~\cite{BN98,Ohl02}. 

Throughout the paper, we use $\cV$ as a countably infinite set of
\emph{variables}. 
Let $\cF$ be a \emph{signature}, a finite set of \emph{function symbols}
each of which has its own fixed arity, and $\Arity_{\cF}(\symb{f})$ be
the arity of function symbol $\symb{f}$. 
We often write $\symb{f}/n \in \cF$ instead of ``$\symb{f} \in \cF$ and
$\Arity_\cF(\symb{f}) = n$'', ``$\symb{f} \in \cF$ such that
$\Arity_\cF(\symb{f})=n$'', and so on.
The set of \emph{terms} over $\cF$ and $V$ ($\subseteq \cV$) is denoted
by $T(\cF,V)$, and the set of variables appearing in any of the terms
$t_1,\ldots,t_n$ is denoted by $\Var(t_1,\ldots,t_n)$.
The number of occurrences of a variable $x$ in a term sequence
$t_1,\ldots,t_n$ is denoted by $\Numv{t_1,\ldots,t_n}{x}$. 
A term $t$ is called \emph{ground} if $\Var(t) = \emptyset$.
A term is called \emph{linear} if any variable occurs in the term at
most once, and called \emph{linear w.r.t.\ a variable} if the variable
appears at most once in $t$. 
For a term $t$ and a position $p$ of $t$, the \emph{subterm of $t$ at
$p$} is denoted by $t|_p$.
The function symbol at the \emph{root} position $\varepsilon$ of term
$t$ is denoted by $\Root(t)$.
Given an $n$-hole \emph{context} $C[~]$ with parallel positions
$p_1,\ldots,p_n$, the notation $C[t_1,\ldots,t_n]_{p_1,\ldots,p_n}$ 
represents the term obtained by replacing hole $\Box$ at position $p_i$
with term $t_i$ for all $1 \leq i \leq n$. 
We may omit the subscript ``$p_1,\ldots,p_n$'' from
$C[\ldots]_{p_1,\ldots,p_n}$. 
For positions $p$ and $p'$ of a term, we write $p' \geq p$ if $p$ is a
prefix of $p'$ (i.e., there exists a sequence $q$ such that $pq = p'$).
Moreover, we write $p' > p$ if $p$ is a proper prefix of $p'$.

A \emph{substitution} $\sigma$ is a mapping from variables to terms such
that the number of variables $x$ with $\sigma(x)\ne x$ is finite, and is
naturally extended over terms.
The \emph{domain} and \emph{range} of $\sigma$ are denoted by
$\Dom(\sigma)$ and $\Range(\sigma)$, respectively.
We may denote $\sigma$ by $\{ x_1 \mapsto t_1, ~ \ldots, ~ x_n \mapsto
t_n \}$ if $\Dom(\sigma) = \{x_1,\ldots,x_n\}$ and $\sigma(x_i) = t_i$
for all $1 \leq i \leq n$.
For $\cF$ and $V$ ($\subseteq \cV$), the set of \emph{substitutions}
that range over $\cF$ and $V$ is denoted by $\Subst(\cF,V)$: 
$\Subst(\cF,V) = \{ \sigma \mid \Range(\sigma) \subseteq T(\cF,V) \}$.
For a substitution $\sigma$ and a term $t$, the application $\sigma(t)$
of $\sigma$ to $t$ is abbreviated to $t\sigma$, and $t\sigma$ is called
an \emph{instance} of $t$.
Given a set $X$ of variables, $\sigma|_X$ denotes the \emph{restricted}
substitution of $\sigma$ w.r.t.\ $X$:
$\sigma|_X = \{ x \mapsto x\sigma \mid x \in \Dom(\sigma) \cap X\}$.

An (oriented) \emph{conditional rewrite rule} over a signature $\cF$ is
a triple $(l,r,c)$, denoted by $l \to r \Leftarrow c$, such that the
\emph{left-hand side} $l$ is a non-variable term in $T(\cF,\cV)$, the
\emph{right-hand side} $r$ is a term in $T(\cF,\cV)$, and the
\emph{conditional part} $c$ is a sequence $\Condition{s_1}{t_1}, \ldots,
\Condition{s_k}{t_k}$ of term pairs ($k \geq 0$) where all of
$s_1,t_1,\ldots,s_k,t_k$ are terms in $T(\cF,\cV)$. 
In particular, a conditional rewrite rule is called \emph{unconditional}
if the conditional part is the empty sequence (i.e., $k = 0$), and we
may abbreviate it to $l \to r$. 
We sometimes attach a unique label $\rho$ to the conditional rewrite
rule $l \to r \Leftarrow c$ by denoting $\rho: l \to r \Leftarrow c$,
and we use the label to refer to the rewrite rule. 

An \emph{(oriented) conditional term rewriting system} (CTRS) over a
signature $\cF$ is a set of conditional rewrite rules over $\cF$.
A CTRS is called an (unconditional) \emph{term rewriting system} (TRS)
if every rule $l \to r \Leftarrow c$ in the CTRS is unconditional and
satisfies $\Var(l) \supseteq \Var(r)$.  
The \emph{reduction relation} $\to_\cR$ of a CTRS $\cR$ is defined as
${\to_{\cR}} = {\bigcup_{n \geq 0} \to_{(n),\cR}}$, where
${\to_{(0),\cR}} = \emptyset$, and ${\to_{(i+1),\cR}} = \{
(C[l\sigma]_p,C[r\sigma]_p) \mid \rho: l \to r \Leftarrow
\Condition{s_1}{t_1}, \ldots, \Condition{s_k}{t_k} \in \cR, ~ 
s_1\sigma \mathrel{\to^*_{(i),\cR}} t_1\sigma, ~\ldots,~ s_k\sigma
\mathrel{\to^*_{(i),\cR}} t_k\sigma \}$ for $i \geq 0$.  
To specify the applied rule $\rho$ and the position $p$ where $\rho$ is
applied, we may write $\to_{p,\rho}$ or $\to_{p,\cR}$ instead of
$\to_\cR$.
Moreover, we may write $\to_{>\varepsilon,\cR}$ instead of $\to_{p,\cR}$
if $p > \varepsilon$. 
The \emph{underlying unconditional system} $\{ l \to r \mid l \to r
\Leftarrow c \in \cR \}$ of $\cR$ is denoted by $\cR_u$.
A term $t$ is called a \emph{normal form} (of $\cR$) if $t$ is
irreducible w.r.t.\ $\cR$. 
For a CTRS $\cR$, a substitution $\sigma$ is called \emph{normalized}
(w.r.t.\ $\cR$) if $x\sigma$ is a normal form w.r.t.\ $\cR$ for every
variable $x \in \Dom(\sigma)$. 
A term $t$ is called \emph{strongly irreducible} (w.r.t.\ $\cR$) if
$t\sigma$ is a normal form w.r.t.\ $\cR$ for every normalized
substitution $\sigma$. 
The sets of \emph{defined symbols} and \emph{constructors} of $\cR$ are
denoted by $\cD_\cR$ and $\cC_\cR$, respectively: 
$\cD_\cR = \{ \Root(l) \mid l \to r \Leftarrow c \in \cR \}$ and
$\cC_\cR = \cF \setminus \cD_\cR$. 
Terms in $T(\cC_\cR,\cV)$ are called \emph{constructor terms of $\cR$}.
$\cR$ is called a \emph{constructor system} if for every rule $l \to r
\Leftarrow c$ in $\cR$, all proper subterms of the $l$ are constructor
terms of $\cR$. 
A CTRS is called \emph{operationally terminating} if there is no
infinite well-formed trees in a certain logical inference
system~\cite{LMM05}.

A conditional rewrite rule $l \to r \Leftarrow c$ is called
\emph{left-linear} (LL) if $l$ is linear, \emph{right-linear} (RL) if
$r$ is linear, \emph{non-erasing} (NE) if $\Var(l) \subseteq \Var(r)$,
and \emph{ground conditional} if $c$ contains no variable.
A conditional rewrite rule $\rho: l \to r \Leftarrow
\Condition{s_1}{t_1}, \ldots, \Condition{s_k}{t_k}$ is called
\emph{weakly-left-linear} (WLL)~\cite{GGS12} if
$\Numv{l,t_1,\dots,t_k}{x} = 1$ for any variable $x \in
\Var(r,s_1,\dots,s_k)$. 
For a syntactic property \textit{P} of conditional rewrite rules, we say
that a CTRS has the property \textit{P} if all of its rules have the
property \textit{P}, e.g., a CTRS is called \emph{LL} if all of its
rules are LL. 
Note that not all LL CTRSs are WLL, e.g., $\symb{f}(x) \to x \Leftarrow
\Condition{\symb{g}(x)}{x}$ is LL but not WLL. 

A conditional rewrite rule $\rho: l \to r \Leftarrow
\Condition{s_1}{t_1}, \ldots, \Condition{s_k}{t_k}$ is called
\emph{deterministic} if $\Var(s_i) \subseteq \Var(l,t_1,\ldots,t_{i-1})$
for all $1 \leq i \leq k$, called \emph{strongly deterministic} if every
term $t_i$ is strongly irreducible w.r.t.\ $\cR$, and called
\emph{syntactically deterministic} if every $t_i$ is a constructor term
or a ground normal form of $\cR_u$. 
We simply call a deterministic CTRS a \emph{DCTRS}, and call a  strongly
or syntactically deterministic CTRS an \emph{SDCTRS}. 
In addition, $\rho$ is classified according to the distribution of
variables in $\rho$ as follows:
\emph{Type 1} if $\Var(r,s_1,t_1,\ldots,s_k,t_k) \subseteq \Var(l)$;
\emph{Type 2} if $\Var(r) \subseteq \Var(l)$;
\emph{Type 3} if $\Var(r) \subseteq \Var(l,s_1,t_1,\ldots,s_k,t_k)$;
\emph{Type 4} otherwise. 
A (D)CTRS is called an \emph{{\it i}-(D)CTRS} if all of its rules are
of Type {\it i}.
A DCTRS $\cR$ is called \emph{normal} (or a \emph{normal CTRS}) if, for
every rule $l \to r \Leftarrow \Condition{s_1}{t_1},\ldots,
\Condition{s_k}{t_k}\in \cR$, all of $t_1,\ldots,t_k$ are ground normal
forms w.r.t.\ $\cR_u$.
In this paper, we only consider 3-DCTRSs.

We often denote a term sequence $t_i,t_{i+1},\ldots,t_{j}$ by
$\MyVec{t}{i}{j}$.
Moreover, for the application of a mapping $\tau$ to $\MyVec{t}{i}{j}$,
we denote the sequence $\tau(t_i),\ldots,\tau(t_j)$ by
$\MapVec{\tau}{t}{i}{j}$, e.g., for a substitution $\theta$, we denote
$t_i\theta,\ldots,t_j\theta$ by $\MapVec{\theta}{t}{i}{j}$. 
For a finite set $X = \{o_1,o_2,\ldots,o_n\}$ of objects, a sequence
$o_1,o_2,\ldots,o_n$ under some arbitrary but fixed order on the objects
is denoted by $\overrightarrow{X}$, and given a mapping $\tau$, the
sequence $\tau(o_1),\tau(o_2),\ldots,\tau(o_n)$ is denoted by
$\tau(\overrightarrow{X})$. 
Given an object $o$, we denote the sequence $\overbrace{o,\ldots,o\,}^{n}$
by $o^n$.

\section{Transformations from DCTRSs into TRSs}
\label{sec:transformations}

In this section, we first recall \emph{soundness} and \emph{completeness} of
transformations, the \emph{simultaneous unraveling}~\cite{Ohl02}, and
the {\SRtransformation}~\cite{SR06} for DCTRSs.
Then, we show that every WLL
DCTRS can be converted to an equivalent WLL and ultra-WLL DCTRS. 
In the following, we use the terminology ``conditional'' for a rewrite
rule that has at least one condition, and distinguish ``conditional
rules'' and ``unconditional rules''. 

\subsection{Soundness and Completeness between Two Rewriting Systems}

We first show a general notion of soundness and completeness between two
(C)TRSs (see~\cite{GG08,NSS12lmcs}). 
We usually consider that one is obtained by transforming the other. 
Let $\cR_1$ and $\cR_2$ be (C)TRSs over signature $\cF_1$ and $\cF_2$,
respectively, $\phi$ be an \emph{initialization} (total) mapping from
$T(\cF_1,\cV)$ to $T(\cF_2,\cV)$, and $\psi$ be a partial inverse of
$\phi$, a so-called \emph{backtranslation} mapping from $T(\cF_2,\cV)$
to $T(\cF_1,\cV)$ such that $\psi(\phi(t_1)) = t_1$ for any term $t_1
\in T(\cF_1,\cV)$. 
We say that
 \begin{itemize}
  \item $\cR_2$ is \emph{sound for} (\emph{reduction of}) \emph{$\cR_1$
	w.r.t.\ $(\phi,\psi)$} if, for any term $t_1 \in T(\cF_1,\cV)$
	and for any term $t_2 \in T(\cF_2,\cV)$, $\phi(t_1)
	\mathrel{\to^*_{\cR_2}} t_2$ implies $t_1
	\mathrel{\to^*_{\cR_1}} \psi(t_2)$ whenever $\psi(t_2)$ is
	defined, 
	and 
  \item $\cR_2$ is \emph{complete for} (\emph{reduction of})
	\emph{$\cR_1$ w.r.t.\ $\phi$} if for all terms $t_1$ and $t'_1$
	in $T(\cF_1,\cV)$, $t_1 \mathrel{\to^*_{\cR_1}} t'_1$ implies
	$\phi(t_1) \mathrel{\to^*_{\cR_2}} \phi(t'_1)$.
 \end{itemize}
We now suppose that $\cR_1$ is a CTRS and $\cR_2$ is a TRS. 
$\cR_2$ is called \emph{computationally equivalent to $\cR_1$} if for
every $\cR_1$-operationally-terminating term $t$ in $T(\cF_1,\cV)$ with
a unique normal form $u$ (i.e., $t \mathrel{\to^*_{\cR_1}} u$), the term
$\phi(t)$ is terminating w.r.t.\ $\cR_2$ and all the normal forms of
$\phi(t)$ w.r.t.\ $\cR_2$ are translated by $\psi$ to $u$.
Note that if $\cR_1$ is operationally terminating, $\cR_2$ is
confluent, terminating, and sound for $\cR_1$ w.r.t.\ $(\phi,\psi)$, and
$\psi$ is defined for all normal forms $t$ such that $\phi(s)
\mathrel{\to^*_{\cR_2}} t$ for some $s \in T(\cF_1,\cV)$, then $\cR_2$
is computationally equivalent to $\cR_1$.

\subsection{Simultaneous Unraveling}

A transformation $U$ of CTRSs into TRSs is called an \emph{unraveling}
if for every CTRS $\cR$, we have that ${\to_\cR} \subseteq
{\to^*_{U(\cR)}}$ and $U(\cR \cup \cR')=U(\cR) \cup \cR'$ whenever
$\cR'$ is a TRS~\cite{Mar96,NSS12lmcs}.
The \emph{simultaneous unraveling} for DCTRSs has been defined
in~\cite{Mar97}, and then has been refined by Ohlebusch~\cite{Ohl01} as
follows. 
\begin{definition}[$\U$~\cite{Ohl02}]
 \label{def:U}
 Let $\cR$ be a DCTRS over a signature $\cF$.
 For each conditional rule $\rho: l \to r \Leftarrow
 \Condition{s_1}{t_1}, \ldots, \Condition{s_k}{t_k}$ in $\cR$, we
 introduce $k$ new function symbols
 $\Usymb^\rho_1,\ldots,\Usymb^\rho_k$, and transform $\rho$ into a set
 of $k+1$ unconditional rules as follows:
 \[
 \U(\rho)
 =
 \{
 ~~~~
 l \to \Usymb^\rho_1(s_1,\overrightarrow{X_1}),
 ~~~~
 \Usymb^\rho_1(t_1,\overrightarrow{X_1}) \to \Usymb^\rho_2(s_2,\overrightarrow{X_2}), 
 ~~~~
 \dots ,
 ~~~~
 \Usymb^\rho_k(t_k,\overrightarrow{X_k}) \to r
 ~~~~
 \}
 \]
 where $X_i = \Var(l,\MyVec{t}{1}{i-1})$ for $1\leq i \leq k$. 
  We define $\U$ for an unconditional rule $l \to r \in \cR$ as $\U(l
 \to r) = \{ l \to r \}$.
 $\U$ is straightforwardly extended to DCTRSs:
 $\U(\cR) = {\bigcup_{\rho \in \cR} \U(\rho)}$.
 We abuse $\U$ to represent the extended signature of $\cF$:
 $\U_\cR(\cF) = \cF \cup \{ \Usymb^\rho_i \mid \rho: l \to r \Leftarrow
 \Condition{s_1}{t_1}, \ldots, \Condition{s_k}{t_k} \in \cR, ~ 1\leq i
 \leq k \}$.
 We say that $\U$ (and also $\U(\cR)$) is \emph{sound for $\cR$} if\/
 $\U(\cR)$ is sound for $\cR$ w.r.t.\ $(\ID_\cF,\ID_{\U_\cR(\cF)})$,
 where $\ID_\cF$ is the identity mapping for $T(\cF,\cV)$, and
 $\ID_{\U_\cR(\cF)}$ is the partial identity mapping for
 $T(\U_\cR(\cF),\cV)$, i.e., $\ID_\cF(t)=\ID_{\U_\cR(\cF)}(t)=t$ for $t
 \in T(\cF,\cV)$ and $\ID_{\U_\cR(\cF)}(t)$ is undefined for $t \in
 T(\U_\cR(\cF),\cV)\setminus T(\cF,\cV)$. 
 We also say that $\U$ (and also $\U(\cR)$) is \emph{complete for $\cR$}
 if\/ $\U(\cR)$ is complete for $\cR$ w.r.t.\ $\ID_\cF$.
\end{definition}
Note that $\U(\cR)$ is a TRS over $\U_\cR(\cF)$, i.e., $\U$ transforms a
DCTRS into a TRS. 
In examples below, we use $\symb{u}_1, \symb{u}_2,\ldots$ for fresh U
symbols introduced during the application of $\U$.

\begin{example}
 \label{ex:qsort}
 Consider the following DCTRS from~\cite[Section~7.2.5]{Ohl02}:
 \[
 \cRqsort = 
 \left\{
 \begin{array}{r@{\>}c@{\>}l}
  \symb{split}(x,\symb{nil}) & \to & \symb{pair}(\symb{nil},\symb{nil}), \\
  \symb{split}(x,\symb{cons}(y,ys)) & \to & \symb{pair}(xs,\symb{cons}(y,zs))
   \Leftarrow \Condition{\symb{split}(x,ys)}{\symb{pair}(xs,zs)},~
   \Condition{x \leq y}{\symb{true}}, \\
  \symb{split}(x,\symb{cons}(y,ys)) & \to & \symb{pair}(\symb{cons}(y,xs),zs)
   \Leftarrow \Condition{\symb{split}(x,ys)}{\symb{pair}(xs,zs)},~
   \Condition{x \leq y}{\symb{false}}, \\
  \symb{qsort}(\symb{nil}) & \to & \symb{nil}, \\
  \symb{qsort}(\symb{cons}(x,xs)) & \to & \symb{qsort}(ys) \Append \symb{cons}(x,\symb{qsort}(zs))
   \Leftarrow \Condition{\symb{split}(x,xs)}{\symb{pair}(ys,zs)} \\
 \end{array}
 \right\}
 \cup \cRsub
 \]
 where 
 \[
 \cRsub = 
 \left\{
 \begin{array}{r@{\>}c@{\>}l@{~~~~~}r@{\>}c@{\>}l@{~~~~~}r@{\>}c@{\>}l}
  \symb{0} \leq y & \to & \symb{true}, &
   \symb{s}(x) \leq \symb{0} & \to & \symb{false}, &
   \symb{s}(x) \leq \symb{s}(y) & \to & x \leq y, \\
  \symb{nil} \Append ys & \to & ys, &
   \symb{cons}(x,xs) \Append ys & \to & \symb{cons}(x, xs \Append ys) \\
 \end{array}
 \right\}
 \]
 Introducing U symbols $\UsplitA$, $\UsplitB$, $\UsplitC$, $\UsplitD$,
 and $\Uqsort$ for conditional rules in $\cRqsort$, $\cRqsort$ is
 unraveled by $\U$ as follows:
 \[
 \U(\cRqsort) = 
 \left\{
 \begin{array}{r@{\>}c@{\>}l}
  \symb{split}(x,\symb{nil}) & \to & \symb{pair}(\symb{nil},\symb{nil}), \\
  
    \symb{split}(x,\symb{cons}(y,ys)) & \to & 
    \UsplitA(\symb{split}(x,ys),x,y,ys), \\
  \UsplitA(\symb{pair}(xs,zs),x,y,ys) & \to & \UsplitB(x \leq y,x,y,ys,xs,zs), \\
  \UsplitB(\symb{true},x,y,ys,xs,zs) & \to & \symb{pair}(xs,\symb{cons}(y,zs)), \\
  
    \symb{split}(x,\symb{cons}(y,ys)) & \to & 
    \UsplitC(\symb{split}(x,ys),x,y,ys), \\
  \UsplitC(\symb{pair}(xs,zs),x,y,ys) & \to & \UsplitD(x \leq y,x,y,ys,xs,zs), \\
  \UsplitD(\symb{false},x,y,ys,xs,zs) & \to & \symb{pair}(\symb{cons}(y,xs),zs), \\
  
    \symb{qsort}(\symb{nil}) & \to & \symb{nil}, \\
  
    \symb{qsort}(\symb{cons}(x,xs)) & \to & 
    \Uqsort(\symb{split}(x,xs),x,xs), \\
  \Uqsort(\symb{pair}(ys,zs),x,xs) & \to & \symb{qsort}(ys) \Append \symb{cons}(x,\symb{qsort}(zs)) \\
 \end{array}
 \right\}
 \cup \cRsub
 \]
\end{example}

As shown in~\cite{Mar96,GGS12}, $\U$ is not sound for all DCTRSs, while
$\U$ is sound for some classes of DCTRSs, e.g., ``confluent and
right-stable'', ``WLL'', and ``RL'' (cf.~\cite{GGS12}).
\begin{theorem}[\cite{GGS12}]
 \label{th:U-soundness_for_WLL}
 $\U$ is sound for WLL DCTRSs.
\end{theorem}

Let \textit{P} be a property on rewrite rules, and $U$ be an
unraveling. A conditional rewrite rule $\rho$ is said to be
\emph{ultra-\textit{P} w.r.t.\ $U$} ($U$-\textit{P}, for short) if all
the rules in $U(\rho)$ have the property \textit{P}.
Note that $U$-\textit{P} is a syntactic property on rewrite rules, and
thus a DCTRS is called \emph{$U$-\textit{P}} if all rules in the DCTRS
are $U$-\textit{P}.
For example, $\cR$ is $\U$-LL if $\U(\cR)$ is LL. 
Some ultra-properties are reformulated without referring to unraveled
systems (cf.~\cite{NSS12lmcs}).
In addition, by definition, the $\U$-WLL property is characterized
without $\U$ as follows.
\begin{theorem}
 \label{thm:characterization_of_U-WLL}
$\cR$ is $\U$-WLL if and only if all unconditional rules in $\cR$ are
 WLL and every conditional rule $l \to r \Leftarrow
 \Condition{s_1}{t_1},\ldots,\Condition{s_k}{t_k}$ ($k>0$) in $\cR$
 satisfies that 
 \begin{itemize}
  \item[(a)] the sequence $l,t_1,\ldots,t_{k-1}$ is linear,
	     and
  \item[(b)] $\Numv{l,t_1,\ldots,t_k}{x} \leq 1$ for any variable $x \in
	     \Var(r)$. 
 \end{itemize}
\end{theorem}
Note that every $\U$-LL DCTRS is $\U$-WLL, while the converse of this
implication does not hold in general.
On the other hand, the class of $\U$-WLL DCTRSs is incomparable with the
class of WLL DCTRSs, e.g., $\symb{f}(x) \to x  \Leftarrow
\Condition{\symb{a}}{y},~\Condition{\symb{b}}{y},~\Condition{x}{\symb{c}}$
is WLL but not $\U$-WLL, and $\symb{f}(x) \to x  \Leftarrow
\Condition{\symb{a}}{y},~\Condition{y}{\symb{b}},~\Condition{\symb{c}}{y}$
is $\U$-WLL but not WLL. 
Though, every WLL DCTRS can be converted to a WLL and $\U$-WLL DCTRS
such that the reductions of these DCTRSs are the same.

In the following, we show that every WLL DCTRS $\cR$ can be converted to
an equivalent WLL and $\U$-WLL DCTRS.
We first convert a WLL conditional rule $\rho: l \to r \Leftarrow
\Condition{s_1}{t_1},\ldots,\Condition{s_k}{t_k}$ to a WLL and $\U$-WLL
one as follows: 
for every variable $x$ in $\rho$ such that
$\Numv{l,t_1,\ldots,t_k}{x}>1$, we linearize the occurrences of $x$ by
replacing each of them by a fresh variable, obtaining $\rho': l' \to r
\Leftarrow \Condition{s_1}{t'_1},\ldots,\Condition{s_k}{t'_k}$;%
\footnote{ Such $x$ does not appear in any of $r,s_1,\ldots,s_k$ because
$\rho$ is WLL.} 
Let $x_1,\ldots,x_j$ be the introduced variables, and $\sigma$ be the
variable renaming that maps $x_i$ to the original one, i.e.,
$\Dom(\sigma)=\{x_1,\ldots,x_j\}$, $l'\sigma = l$, and $t'_i\sigma =
t_i$ for $1 \leq i \leq k$;
We add the condition
$\Condition{\Tuple_j(x_1,\ldots,x_j)}{\Tuple_j(x_1\sigma,\ldots,x_j\sigma)}$
into $\rho'$ as the last condition, where $\Tuple_j$ is a fresh $j$-ary
constructor. 
We denote this transformation by $\T$, i.e., $\T(\rho)= l' \to r
\Leftarrow
\Condition{s_1}{t'_1},\ldots,\Condition{s_k}{t'_k},\Condition{\Tuple_j(x_1,\ldots,x_j)}{\Tuple_j(x_1\sigma,\ldots,x_j\sigma)}$. 
In addition, we abuse $\T$ for unconditional rules and $\cR$: $\T(l \to
r) = l \to r$ and $\T(\cR)= \{ \T(\rho) \mid \rho \in \cR \}$. 
By definition, $\T(\rho)$ is WLL and $\U$-WLL, i.e., $\T$ transforms a
WLL DCTRS into a WLL and $\U$-WLL DCTRS. 
It is clear that if $s \in T(\cF,\cV)$ and $s \mathrel{\to^*_{\T(\cR)}}
t$, then $t \in T(\cF,\cV)$.
\begin{theorem}
 \label{thm:from_WLL_to_UWLL}
 Let $\cR$ be a WLL DCTRS over a signature $\cF$.
 Then, 
 ${\to^*_\cR} = {\to^*_{\T(\cR)}}$ over $T(\cF,\cV)$
\end{theorem}
\begin{proofsketch}
The following two claims can be proved by induction on the lexicographic
 product $(m,n)$: 
(i) if $s \mathrel{\to^n_{(m),\cR}} t$ then $s \mathrel{\to^*_{\T(\cR)}}
 t$, and 
(ii) if $s \mathrel{\to^n_{(m),\T(\cR)}} t$ then $s \mathrel{\to^*_\cR}
 t$. 
\qed
\end{proofsketch}

\subsection{The {\SRTransformation}}
\label{subsec:SR}

Next, we introduce the {\SRtransformation} and its properties.
Before transforming a CTRS $\cR$, we first extend the signature of $\cR$
as follows:
\begin{itemize}
 \item we keep the constructors of $\cR$, while replacing $\symb{c}/n$
       by $\olsymb{c}/n$, 
 \item the arity $n$ of defined symbol $\symb{f}$ is extended to $n+m$
       where $\symb{f}$ has $m$ conditional rules in $\cR$, replacing
       $\symb{f}$ by $\olsymb{f}$, the arity of which is $n+m$, 
 \item a fresh constant $\bot$ and a fresh unary symbol $\guard{\cdot}$
       are introduced,
       and
 \item for every conditional rule $\rho: l \to r \Leftarrow
       \Condition{s_1}{t_1}, \ldots, \Condition{s_k}{t_k}$ in $\cR$, we
       introduce $k$ fresh symbols $[\ ]^\rho_1,[\ ]^\rho_2,\ldots,[\
       ]^\rho_k$ with the arities $1, 1+|\Var(t_1)|, 1+|\Var(t_1,t_2)|,
       \ldots, 1+|\Var(t_1,\ldots,t_{k-1})|$.
\end{itemize}
We assume that for every defined symbol $\symb{f}$, the conditional
rules for $\symb{f}$ are ranked by some arbitrary but fixed order.
We denote the extended signature by $\ol{\cF}$:
$\ol{\cF} = 
\{ \olsymb{c} \mid \symb{c} \in \cC_\cR\}
\cup
\{ \olsymb{f} \mid \symb{f} \in \cD_\cR \}
\cup
\{ \bot, \guard{\cdot}\}
\cup
\{ [\ ]^\rho_j \mid \rho: l \to r \Leftarrow
       \Condition{s_1}{t_1}, \ldots, \Condition{s_k}{t_k} \in \cR, ~ 1 \leq j \leq k \}$.
We introduce a mapping $\Extension$ to extend the arguments of defined
symbols in a term as follows: 
$\Extension(x) = x$ for $x \in \cV$; 
$\Extension(\symb{c}(\MyVec{t}{1}{n})) =
\symb{c}(\MapVec{\Extension}{t}{1}{n})$ for $\olsymb{c}/n \in \cC_\cR$;
$\Extension(\symb{f}(\MyVec{t}{1}{n})) =
\olsymb{f}(\MapVec{\Extension}{t}{1}{n},\MyVec{z}{1}{m})$ for
$\symb{f}/n \in \cD_\cR$, where $\symb{f}$ has $m$
conditional rules in $\cR$, $\Arity_{\ol{\cF}}(\olsymb{f}) = n+m$, and
$z_1,\ldots,z_m$ are fresh variables.
The extended arguments of $\olsymb{f}$ are used for evaluating the
corresponding conditions, and the fresh constant $\bot$ is introduced to
the extended arguments of defined symbols, which does not store any
evaluation.
To put $\bot$ into the extended arguments, we define a mapping\/
$\Reset{\cdot}$ that puts $\bot$ to all the extended arguments of
defined symbols, as follows:
$\Reset{x} = x$ for $x \in \cV$; 
$\Reset{\olsymb{c}(\MyVec{t}{1}{n})} =
\olsymb{c}(\Reset{\MyVec{t}{1}{n}})$ for 
$\symb{c}/n \in \cC_\cR$; 
$\Reset{\olsymb{f}(\MyVec{t}{1}{n},\MyVec{u}{1}{m})} =
\olsymb{f}(\Reset{\MyVec{t}{1}{n}},\bot,\ldots,\bot)$ for 
$\symb{f}/n \in \cD_\cR$; 
$\Reset{\guard{t}} = \guard{\Reset{t}}$; 
$\Reset{\bot} = \bot$; 
$\Reset{[\ldots]^\rho_j}=\bot$. 
Now we define a mapping {$\ol{\cdot}$} from $T(\cF,\cV)$ to
$T(\ol{\cF},\cV)$ as $\ol{t} = \Reset{\Extension(t)}$.
On the other hand, the partial inverse mapping $\,\wh{\cdot}\,$ for
$\,\ol{\cdot}\,$ is defined as follows: 
$\wh{x}=x$ for $x \in \cV$;
$\wh{\olsymb{c}(\MyVec{t}{1}{n})}=\symb{c}(\wh{t_1},\ldots,\wh{t_n})$
for $\symb{c}/n \in \cC_\cR$;
$\wh{\olsymb{f}(\MyVec{t}{1}{n},\ldots)}=\symb{f}(\wh{t_1},\ldots,\wh{t_n})$
for $\symb{f}/n \in \cD_\cR$;
$\wh{\guard{t}}=\,\wh{t}\,$.
Note that in applying $\Reset{\cdot}$ or $\,\wh{\cdot}\,$ to
\emph{reachable terms} defined later, the case of applying
$\Reset{\cdot}$ to $\bot$ or $[\ldots]^\rho_j$ never happens.

The {\SRtransformation}~\cite{SR06b} for SDCTRSs has been defined for
only $\U$-LL SDCTRSs---more precisely, any other case has not been
discussed in~\cite{SR06b}.
Originally, to generate a computationally equivalent TRS, a given CTRS
$\cR$ is assumed to be a $\U$-LL SDCTRS, while such an assumption is a
sufficient condition for computational equivalence.
To define the transformation itself, $\cR$ does not have to be an
SDCTRS, but the $\U$-LL property is used to ensure that for $\rho: l \to
r \Leftarrow \Condition{s_1}{t_1}, \ldots, \Condition{s_k}{t_k}$, the
sequence $l,t_1,\ldots,t_{k-1}$ is linear. 
To ensure it, the $\U$-WLL property is enough because of
Theorem~\ref{thm:characterization_of_U-WLL}~(a).
For this reason, the {\SRtransformation} is applicable not only to $\U$-LL
SDCTRSs but also to $\U$-WLL DCTRSs without any change.
\begin{definition}[$\SR$~\cite{SR06b}]
 \label{def:SR}
 Let $\cR$ be a $\U$-WLL DCTRS over a signature $\cF$ and $\ol{\cF}$ be
 the extended signature of $\cF$ mentioned above. 
 Then, the $i$-th conditional $\symb{f}$-rule $\rho:
 \symb{f}(\MyVec{w}{1}{n}) \to r \Leftarrow \Condition{s_1}{t_1},
 \ldots, \Condition{s_k}{t_k}$ is transformed into a set of $k+1$
 unconditional rules as follows: 
 \[
 \SR(\rho)=
 \left\{
 \begin{array}{r@{\,}c@{\,}r@{\>}c@{\>}l@{\,}c@{\,}l}
  \olsymb{f}(\MyVec{w'}{1}{n},\MyVec{z}{1}{i-1}, & \bot, & \MyVec{z}{i+1}{m})
   & \to & 
   \olsymb{f}(\MyVec{w'}{1}{n},\MyVec{z}{1}{i-1}, & [\guard{\ol{s_1}},\overrightarrow{V_1}]_1^\rho, & \MyVec{z}{i+1}{m}), \\[3pt]
   \olsymb{f}(\MyVec{w'}{1}{n},\MyVec{z}{1}{i-1}, & [\guard{\Extension(t_1)},\overrightarrow{V_1}]_1^\rho, & \MyVec{z}{i+1}{m})
   & \to & 
   \olsymb{f}(\MyVec{w'}{1}{n},\MyVec{z}{1}{i-1}, & [\guard{\ol{s_2}},\overrightarrow{V_2}]_2^\rho, & \MyVec{z}{i+1}{m}), \\
   & & & \vdots &  \\
  \olsymb{f}(\MyVec{w'}{1}{n},\MyVec{z}{1}{i-1}, & [\guard{\Extension(t_k)},\overrightarrow{V_k}]_k^\rho, & \MyVec{z}{i+1}{m})
   & \to &
   \guard{\ol{r}}
 \end{array}
 \right\}
 \]
 where $\MyVec{w'}{1}{n} = \MapVec{\Extension}{w}{1}{n}$, $V_j = \Var(\MyVec{t}{1}{j-1})$ for all $1 \leq j \leq k$,%
 \footnote{ For arbitrary DCTRSs, we may define $V_j$ as $V_j =
 \Var(\MyVec{t}{1}{j-1})\setminus\Var(\MyVec{w}{1}{n})$.} 
 and  
 $z_1,\ldots,z_{i-1},z_{i+1},\ldots,z_m$ are fresh variables. 
 An unconditional rule in $\cR$ is converted as follows:
 $\SR(l \to r) = \{~ \Extension(l) \to \guard{\ol{r}} ~\}$, 
 that is, $\SR(\symb{f}(\MyVec{w}{1}{n}) \to r) = \{~
 \olsymb{f}(\MapVec{\Extension}{w}{1}{n},\MyVec{z}{1}{m}) \to
 \guard{\ol{r}}~\}$, where $z_1,\ldots,z_m$ are fresh variables. 
 The set of auxiliary rules is defined as follows:
 \[
 \begin{array}{@{}l@{\>}l@{}}
  \cRaux
   = &
   \{~ 
   \guard{\guard{x}} \to \guard{x} 
   ~\} 
   \cup
   \{~
   \olsymb{c}(\MyVec{x}{1}{i-1},\guard{x_i},\MyVec{x}{i+1}{n}) 
   \to  \guard{\olsymb{c}(\MyVec{x}{1}{n})} \mid \symb{c}/n \in \cC_\cR,
   ~ 1 \leq i \leq n
   ~\} \\[3pt]
   &
   {}\cup
   \{~
   \olsymb{f}(\MyVec{x}{1}{i-1},\guard{x_i},\MyVec{x}{i+1}{n},\MyVec{z}{1}{m})
   \to  \guard{\olsymb{f}(\MyVec{x}{1}{n},\bot,\ldots,\bot)} \mid
   \symb{f}/n \in \cD_\cR, 
   ~1 \leq i \leq n
   ~\}
   \\
 \end{array}
\]
 where $x_1,\ldots,x_n,z_1,\ldots,z_m$ are distinct variables.  
 The transformation $\SR$ is defined as follows: 
 $\SR(\cR) = \bigcup_{\rho \in \cR} \SR(\rho) \cup \cRaux$.
 We say that $\SR$ (and also $\SR(\cR)$) is \emph{sound for $\cR$} if\/
 $\SR(\cR)$ is sound for $\cR$ w.r.t.\
 $(\,\guard{\,\ol{\cdot}\,},\,\wh{\cdot}\,)$.
 We also say that $\SR$ (and also $\SR(\cR)$) is \emph{complete for
 $\cR$} if\/ $\SR(\cR)$ is complete for $\cR$ w.r.t.\ 
 $\guard{\,\ol{\cdot}\,}$.
\end{definition}
Note that $\SR(\cR)$ is a TRS over $\ol{\cF}$, i.e., $\SR$ transforms a
$\U$-WLL DCTRS into a TRS. 
In examples below, we use $[~]_1, [~]_2,\ldots$ for fresh tuple symbols
introduced during the application of $\SR$, and we may abuse $\symb{f}$
instead of $\olsymb{f}$ if all the rules for $\symb{f}$ in $\cR$ are
unconditional, and as in~\cite{SR06,SR06b}, the original constructor
$\symb{c}$ is abused instead of $\olsymb{c}$.
It has been shown in~\cite{SR06} that $\SR$ is complete for all $\U$-LL
SDCTRSs. 
By definition, it is clear that $\SR$ is also complete for all $\U$-WLL
DCTRSs.
\begin{theorem}
 \label{thm:completeness_of_SR}
 $\SR$ is complete for $\U$-WLL DCTRSs.
\end{theorem}

To evaluate conditions of the $i$-th conditional rule
$\symb{f}(\MyVec{w}{1}{n}) \to r_i \Leftarrow \Condition{s_{1}}{t_{1}},
\ldots, \Condition{s_{k}}{t_{k}}$, the $i$-th conditional rule is
transformed into the $k+1$ unconditional rules: 
a term of the form $[\guard{t},\MyVec{u}{1}{n_j}]^\rho_j$ represents an
intermediate state $t$ of the evaluation of the $j$-th condition
$\Condition{s_j}{t_j}$ carrying $\MyVec{u}{1}{n_j}$ for
$\Var(\MyVec{t}{1}{j-1})$, the first unconditional rule starts to
evaluate the condition (an instance of $\ol{s_1}$), and the remaining
$k$ rules examine whether the corresponding conditions hold. 
On the other hand, the first rule $\guard{\guard{x}} \to \guard{x}$ in
$\cRaux$ removes the nesting of $\guard{\cdot}$, the second rule
$\olsymb{c}(\MyVec{x}{1}{i-1},\guard{x_i},\MyVec{x}{i+1}{n}) \to
\guard{\olsymb{c}(\MyVec{x}{1}{n})}$ is used for shifting
$\guard{\cdot}$ upward, and the third rule
$\olsymb{f}(\MyVec{x}{1}{i-1},\guard{x_i},\MyVec{x}{i+1}{n},\MyVec{z}{1}{m})
\to \guard{\olsymb{f}(\MyVec{x}{1}{n},\bot,\ldots,\bot)}$ is used for
both shifting $\guard{\cdot}$ upward and resetting the evaluation of
conditions at the extended arguments of $\olsymb{f}$. 
The unary symbol $\guard{\cdot}$ and its rules in $\cRaux$ are
introduced to preserve confluence of the original CTRS $\cR$ on
reachable terms (see~\cite{SR06} for the detail of the role of
$\guard{\cdot}$ and its rules). 

\begin{example}
 \label{ex:qsort-SR}
 Consider $\cRqsort$ in Example~\ref{ex:qsort} again.
 Introducing tuple symbols $\EvalsplitA{~}$,
 $\EvalsplitB{~}$,$\EvalsplitC{~}$, $\EvalsplitD{~}$, and
 $\Evalqsort{~}$, $\cRqsort$ is transformed by $\SR$ as follows: 
 \[
 \begin{array}{l@{\>}l}
  \SR(\cRqsort) = &
   \left\{
    \begin{array}{r@{\>}c@{\>}l}
     \olsymb{split}(x,\symb{nil},z_1,z_2) & \to & \guard{\symb{pair}(\symb{nil},\symb{nil})} \\
     \olsymb{split}(x,\symb{cons}(y,ys),\bot,z_2) & \to & 
      \olsymb{split}(x,\symb{cons}(y,ys),\EvalsplitA{\guard{\olsymb{split}(x,ys,\bot,\bot)}},z_2) \\
     \olsymb{split}(x,\symb{cons}(y,ys),\EvalsplitA{\guard{\symb{pair}(xs,zs)}},z_2) & \to & 
      \olsymb{split}(x,\symb{cons}(y,ys),\EvalsplitB{\guard{x \leq y},xs,zs},z_2) \\
     \olsymb{split}(x,\symb{cons}(y,ys),\EvalsplitB{\guard{\symb{true}},xs,zs},z_2) & \to & \guard{\symb{pair}(xs,\symb{cons}(y,zs))} \\
     \olsymb{split}(x,\symb{cons}(y,ys),z_1,\bot) & \to & 
      \olsymb{split}(x,\symb{cons}(y,ys),z_1,\EvalsplitC{\guard{\olsymb{split}(x,ys,\bot,\bot)}}) \\
     \olsymb{split}(x,\symb{cons}(y,ys),z_1,\EvalsplitC{\guard{\symb{pair}(xs,zs)}}) & \to & 
      \olsymb{split}(x,\symb{cons}(y,ys),z_1,\EvalsplitD{\guard{x \leq y},xs,zs}) \\
     \olsymb{split}(x,\symb{cons}(y,ys),z_1,\EvalsplitD{\guard{\symb{false}},xs,zs}) & \to & 
      \guard{\symb{pair}(\symb{cons}(y,xs),zs)} \\
     \olsymb{qsort}(\symb{nil},z_1) & \to & \guard{\symb{nil}} \\
     \olsymb{qsort}(\symb{cons}(x,xs),\bot) & \to & 
      \olsymb{qsort}(\symb{cons}(x,xs),\Evalqsort{\guard{\olsymb{split}(x,xs,\bot,\bot)}}) \\
     \olsymb{qsort}(\symb{cons}(x,xs),\Evalqsort{\guard{\symb{pair}(ys,zs)}})
      & \to & \guard{\olsymb{qsort}(ys,\bot) \Append \symb{cons}(x,\olsymb{qsort}(zs,\bot))} \\
    \end{array}
 \right\}
 \\
 & {} \cup \cRqsortaux
 \end{array}
 \]
 where
 \[
 \cRqsortaux =
 \left\{
 \begin{array}{c}
  \begin{array}{r@{\>}c@{\>}l@{~~~~~}r@{\>}c@{\>}l@{~~~~~}r@{\>}c@{\>}l}
   \symb{0} \leq y & \to & \guard{\symb{true}}, &
    \symb{s}(x) \leq \symb{0} & \to & \guard{\symb{false}}, &
    \symb{s}(x) \leq \symb{s}(y) & \to & \guard{x \leq y}, \\
   \symb{nil} \Append ys & \to & \guard{ys}, &
    \symb{cons}(x,xs) \Append ys & \to & \guard{\symb{cons}(x, xs \Append ys)}, \\
  \end{array}
  \\
  \begin{array}{r@{\>}c@{\>}l@{~~~~~~~~}r@{\>}c@{\>}l}
   \guard{\guard{x}} & \to & \guard{x},
    &
    \symb{s}(\guard{x}) & \to & \guard{\symb{s}(x)}, \\
   \symb{cons}(\guard{x},xs) & \to & \guard{\symb{cons}(x,xs)}, &
    \symb{cons}(x,\guard{xs}) & \to & \guard{\symb{cons}(x,xs)}, \\
   \symb{pair}(\guard{x},y) & \to & \guard{\symb{pair}(x,y)}, &
    \symb{pair}(x,\guard{y}) & \to & \guard{\symb{pair}(x,y)}, \\
   \olsymb{split}(\guard{x},ys,z_1,z_2) & \to & \guard{\olsymb{split}(x,ys,\bot,\bot)}, &
    \olsymb{split}(x,\guard{ys},z_1,z_2) & \to & \guard{\olsymb{split}(x,ys,\bot,\bot)}, \\
   \olsymb{qsort}(\guard{xs},z_1) & \to & \guard{\olsymb{qsort}(xs,\bot)}, \\
   \guard{x} \leq y & \to & \guard{x \leq y}, &
    x \leq \guard{y} & \to & \guard{x \leq y}, \\
   \guard{xs} \Append ys & \to & \guard{xs \Append ys}, &
    xs \Append \guard{ys} & \to & \guard{xs \Append ys} \\
  \end{array}
 \end{array}
 \right\}
 \]
 $\cRqsortaux$ is not a constructor system because of e.g.,
 $\guard{\guard{x}} \to \guard{x}$, and thus, $\SR(\cR_1)$ is not a
 constructor system, either, while $\cR_1$ is so.
\end{example}

Rules in $\U(\cR)$ and $\SR(\cR)\setminus \cRaux$ have some
correspondence each other. 
An unconditional rule $l \to r \in \cR \cap \U(\cR)$ is said to
\emph{correspond to} $\Extension(l) \to \guard{\ol{r}} \in \SR(\cR)$,
and vice versa; 
For the $i$-th conditional $\symb{f}$-rule $\rho:
\symb{f}(\MyVec{w}{1}{n}) \to r \Leftarrow
\Condition{s_1}{t_1},\ldots,\Condition{s_k}{t_k}$, the $j$-th rule of
$\U(\rho)$ in Definition~\ref{def:U} is said to \emph{correspond to} the
$j$-th rule of $\SR(\rho)$ in Definition~\ref{def:SR}, and vice versa.

One of the important properties of $\SR$ is that $\U(\cR)$ is WLL if and
only if so is $\SR(\cR)$. 
By definition, this claim holds trivially. 
\begin{theorem}
 \label{thm:UWLL-properties}
 $\cR$ is $\U$-WLL if and only if\/ $\SR(\cR)$ is WLL.
\end{theorem}

A term $t$ in $T(\ol{\cF},\cV)$ is called \emph{reachable} if there
exists a term $s$ in $T(\cF,\cV)$ such that $\guard{\ol{s}}
\mathrel{\to_{\SR(\cR)}^*} t$.
It is clear that for any reachable term $t \in T(\ol{\cF},\cV)$, any
term $t' \in T(\ol{\cF},\cV)$ with $t \mathrel{\to^*_{\SR(\cR)}} t'$ is
reachable.
In the following, for the extended signature $\ol{\cF}$, we only
consider subterms of reachable terms because it suffices to consider
them in discussing soundness. 
For brevity, subterms of reachable terms are also called
\emph{reachable}. 
In reachable terms, the introduced symbols $\bot$ and $[~]^\rho_i$
appear at appropriate positions of a term, i.e., at the root position of
the term or an $i$-th argument of a subterm $\olsymb{f}(\ldots)$ where
$i > n$ and $\symb{f}$ is an $n$-ary defined symbol.

\section{Soundness for WLL and Ultra-WLL DCTRSs}
\label{sec:relative-soundness}

In this section, we prove that $\SR$ is sound for WLL and $\U$-WLL
DCTRSs. 
In the following, we use $\cR$ as a $\U$-WLL DCTRS over a signature
$\cF$.

It would be possible to follow the proof shown in~\cite{NYG14wpte} for
soundness of $\SR$ for WLL normal CTRSs.
However, the proof is very long, and it is easy to guess that an
analogous proof for WLL and $\U$-WLL DCTRSs---more complicated systems
than normal CTRSs---becomes much longer.
In this paper, we try to shorten the proof, providing a clearer one. 

Our insight for a proof is that a term in $T(\ol{\cF},\cV)$
represents some corresponding terms in $T(\U(\cF),\cV)$, and a
derivation of $\SR(\cR)$ starting with $\guard{\ol{s}}$ ($s \in
T(\cF,\cV)$) represents the corresponding computation tree of $\U(\cR)$,
whose root is $s$. 
We illustrate this observation by the following WLL and $\U$-WLL normal
DCTRS:
\[
 \cRidea
 =
 \left\{
 \begin{array}{r@{\>}c@{\>}l@{~~~~~}r@{\>}c@{\>}l@{~~~~~}r@{\>}c@{\>}l@{~~~~~}r@{\>}c@{\>}l}
  \symb{f}(x) & \to & \symb{c} \Leftarrow \Condition{x}{\symb{c}},
   &
   \symb{a} & \to & \symb{c},
   &
   \symb{b} & \to & \symb{c},
   \\
  \symb{f}(x) & \to & \symb{d} \Leftarrow \Condition{x}{\symb{d}},
   &
   \symb{a} & \to & \symb{d},
   &
   \symb{b} & \to & \symb{d},
   \\
  \symb{g}(x) & \to & \symb{h}(x,x),
   &
   \symb{h}(\symb{c},\symb{d}) & \to & \symb{c},
   &
   \symb{h}(x,\symb{f}(x)) & \to & \symb{d} 
   \\
 \end{array}
 \right\}
\]
To simplify the discussion, we use a normal CTRS, and omit $[~]^\rho_j$
introduced during the application of $\SR$.
$\cRidea$ is transformed by $\U$ and $\SR$, respectively, as follows:
\[
 \U(\cRidea)
 =
 \left\{
 \begin{array}{r@{\>}c@{\>}l@{~~~~~}r@{\>}c@{\>}l@{~~~~~}r@{\>}c@{\>}l@{~~~~~}r@{\>}c@{\>}l@{~~~~~}r@{\>}c@{\>}l}
  \symb{f}(x) & \to & \UfC(x,x), &
   \UfC(\symb{c},x) & \to & \symb{c},
   &
   \symb{a} & \to & \symb{c},
   &
   \symb{b} & \to & \symb{c},
   \\
  \symb{f}(x) & \to & \UfD(x,x), &
   \UfD(\symb{d},x) & \to & \symb{d},
   &
   \symb{a} & \to & \symb{d},
   &
   \symb{b} & \to & \symb{d},
   \\
  \symb{g}(x) & \to & \symb{h}(x,x),
   &
   \symb{h}(\symb{c},\symb{d}) & \to & \symb{c},
   &
   \symb{h}(x,\symb{f}(x)) & \to & \symb{d}
   \\
 \end{array}
 \right\}
\]
\[
 \SR(\cRidea)
 =
 \left\{
 \begin{array}{r@{\>}c@{\>}l@{~~~~~}r@{\>}c@{\>}l@{~~~~~}r@{\>}c@{\>}l@{~~~~~}r@{\>}c@{\>}l@{~~~~~}r@{\>}c@{\>}l}
  \olsymb{f}(x,\bot,z_2) & \to & \olsymb{f}(x,\guard{x},z_2), &
   \olsymb{f}(x,\guard{\symb{c}},z_2) & \to & \guard{\symb{c}},
   &
   \olsymb{a} & \to & \guard{\symb{c}},
   &
   \olsymb{b} & \to & \guard{\symb{c}},
   \\
  \olsymb{f}(x,z_1,\bot) & \to & \olsymb{f}(x,z_1,\guard{x}), &
   \olsymb{f}(x,z_1,\guard{\symb{d}}) & \to & \guard{\symb{d}},
   &
   \olsymb{a} & \to & \guard{\symb{d}},
   &
   \olsymb{b} & \to & \guard{\symb{d}},
   \\
  \olsymb{g}(x) & \to & \guard{\olsymb{h}(x,x)},
   &
   \olsymb{h}(\symb{c},\symb{d}) & \to & \guard{\symb{c}},
   &
   \olsymb{h}(x,\olsymb{f}(x,z_1,z_2)) & \to & \guard{\symb{d}},
   & & \ldots
        \\
 \end{array}
 \right\}
\]
Each reachable term in $T(\ol{\cF},\cV)$ represents a finite set of
terms in $T(\U(\cF),\cV)$:
$\olsymb{f}(\olsymb{a},\guard{\symb{c}},\bot)$ represents two terms
$\symb{f}(\symb{a})$ and $\UfC(\symb{c},\symb{a})$; 
$\guard{\olsymb{h}(\olsymb{f}(\olsymb{a},\guard{\symb{c}},\bot),\olsymb{f}(\olsymb{a},\guard{\symb{c}},\bot))}$
represents four terms $\symb{h}(\symb{f}(\symb{a}),\symb{f}(\symb{a}))$,
$\symb{h}(\symb{f}(\symb{a}),\UfC(\symb{c},\symb{a}))$,
$\symb{h}(\symb{f}(\symb{a}),\symb{f}(\symb{a}))$, and
$\symb{h}(\UfC(\symb{c},\symb{a}),\UfC(\symb{c},\symb{a}))$.
These correspondence will be formalized as a mapping $\Convert$ from
$T(\ol{\cF},\cV)$ to $2^{T(\U(\cF),\cV)}$ later. 
Figure~\ref{fig:idea1} illustrates a derivation of $\SR(\cRidea)$ and
its corresponding computation tree (more precisely, a DAG) of
$\U(\cRidea)$ where reduced terms are underlined, and in each row, the
leftmost term is the one appearing in the derivation of $\SR(\cRidea)$
and the remaining are terms in $T(\U(\cF),\cV)$ that are represented by
the leftmost one. 
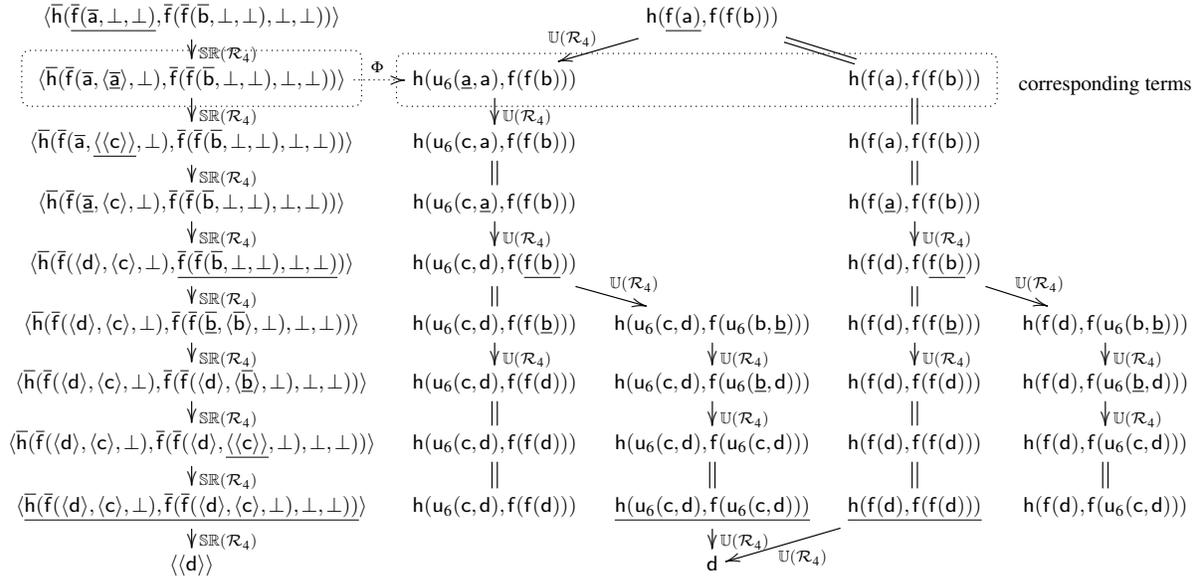
\begin{figure}[t]
{\scriptsize 
\[
    \xymatrix@R=7pt@C=8pt{
        \guard{\olsymb{h}(\underline{\olsymb{f}(\olsymb{a},\bot,\bot)},\olsymb{f}(\olsymb{f}(\olsymb{b},\bot,\bot),\bot,\bot))} \arSR[d] 
        &
        &
        \symb{h}(\underline{\symb{f}(\symb{a})},\symb{f}(\symb{f}(\symb{b}))) \arUl[dl] \ar@{=}[dr]
        &
        \\
        \Enc{2l}
        \guard{\olsymb{h}(\olsymb{f}(\olsymb{a},\guard{\underline{\olsymb{a}}},\bot),\olsymb{f}(\olsymb{f}(\olsymb{b},\bot,\bot),\bot,\bot))} \arSR[d]
        &
        \Enc[rr]{2r}
        \ar@{.>}^<<<<{\Convert} "2,1" ; "Enc2r"
        \symb{h}(\UfC(\underline{\symb{a}},\symb{a}),\symb{f}(\symb{f}(\symb{b}))) \arU
        &
        &
        \symb{h}(\symb{f}(\symb{a}),\symb{f}(\symb{f}(\symb{b}))) \ar@{=}[d] 
        &
        \txt{corresponding terms}
        \\
        \guard{\olsymb{h}(\olsymb{f}(\olsymb{a},\underline{\guard{\guard{\symb{c}}}},\bot),\olsymb{f}(\olsymb{f}(\olsymb{b},\bot,\bot),\bot,\bot))} \arSR[d]
        &
        \symb{h}(\UfC(\symb{c},\symb{a}),\symb{f}(\symb{f}(\symb{b}))) \ar@{=}[d]
        &
        &
        \symb{h}(\symb{f}(\symb{a}),\symb{f}(\symb{f}(\symb{b}))) \ar@{=}[d] 
        \\
        \guard{\olsymb{h}(\olsymb{f}(\underline{\olsymb{a}},\guard{\symb{c}},\bot),\olsymb{f}(\olsymb{f}(\olsymb{b},\bot,\bot),\bot,\bot))} \arSR[d]
        &
        \symb{h}(\UfC(\symb{c},\underline{\symb{a}}),\symb{f}(\symb{f}(\symb{b}))) \arU
        &
        &
        \symb{h}(\symb{f}(\underline{\symb{a}}),\symb{f}(\symb{f}(\symb{b}))) \arU[d] 
        \\
        \guard{\olsymb{h}(\olsymb{f}(\guard{\symb{d}},\guard{\symb{c}},\bot),\underline{\olsymb{f}(\olsymb{f}(\olsymb{b},\bot,\bot),\bot,\bot)})} \arSR[d]
        &
        \symb{h}(\UfC(\symb{c},\symb{d}),\symb{f}(\underline{\symb{f}(\symb{b})})) \ar@{=}[d] \arU[dr]
        &
        &
        \symb{h}(\symb{f}(\symb{d}),\symb{f}(\underline{\symb{f}(\symb{b})})) \ar@{=}[d] \arU[dr] 
        \\
        \guard{\olsymb{h}(\olsymb{f}(\guard{\symb{d}},\guard{\symb{c}},\bot),\olsymb{f}(\olsymb{f}(\underline{\olsymb{b}},\guard{\olsymb{b}},\bot),\bot,\bot))} \arSR[d]
        &
        \symb{h}(\UfC(\symb{c},\symb{d}),\symb{f}(\symb{f}(\underline{\symb{b}}))) \arU 
        &
        \symb{h}(\UfC(\symb{c},\symb{d}),\symb{f}(\UfC(\symb{b},\underline{\symb{b}}))) \arU
        &
        \symb{h}(\symb{f}(\symb{d}),\symb{f}(\symb{f}(\underline{\symb{b}}))) \arU 
        &
        \symb{h}(\symb{f}(\symb{d}),\symb{f}(\UfC(\symb{b},\underline{\symb{b}}))) \arU
        \\
        \guard{\olsymb{h}(\olsymb{f}(\guard{\symb{d}},\guard{\symb{c}},\bot),\olsymb{f}(\olsymb{f}(\guard{\symb{d}},\guard{\underline{\olsymb{b}}},\bot),\bot,\bot))} \arSR[d]
        &
        \symb{h}(\UfC(\symb{c},\symb{d}),\symb{f}(\symb{f}(\symb{d}))) \ar@{=}[d] 
        &
        \symb{h}(\UfC(\symb{c},\symb{d}),\symb{f}(\UfC(\underline{\symb{b}},\symb{d}))) \arU
        &
        \symb{h}(\symb{f}(\symb{d}),\symb{f}(\symb{f}(\symb{d}))) \ar@{=}[d]  
        &
        \symb{h}(\symb{f}(\symb{d}),\symb{f}(\UfC(\underline{\symb{b}},\symb{d}))) \arU
        \\
        \guard{\olsymb{h}(\olsymb{f}(\guard{\symb{d}},\guard{\symb{c}},\bot),\olsymb{f}(\olsymb{f}(\guard{\symb{d}},\underline{\guard{\guard{\symb{c}}}},\bot),\bot,\bot))} \arSR[d]
        &
        \symb{h}(\UfC(\symb{c},\symb{d}),\symb{f}(\symb{f}(\symb{d}))) \ar@{=}[d]  
        &
        \symb{h}(\UfC(\symb{c},\symb{d}),\symb{f}(\UfC(\symb{c},\symb{d}))) \ar@{=}[d]
        &
        \symb{h}(\symb{f}(\symb{d}),\symb{f}(\symb{f}(\symb{d}))) \ar@{=}[d]   
        &
        \symb{h}(\symb{f}(\symb{d}),\symb{f}(\UfC(\symb{c},\symb{d}))) \ar@{=}[d]
        \\
        \guard{\underline{\olsymb{h}(\olsymb{f}(\guard{\symb{d}},\guard{\symb{c}},\bot),\olsymb{f}(\olsymb{f}(\guard{\symb{d}},\guard{\symb{c}},\bot),\bot,\bot))}} \arSR[d]
        &
        \symb{h}(\UfC(\symb{c},\symb{d}),\symb{f}(\symb{f}(\symb{d}))) 
        &
        \underline{\symb{h}(\UfC(\symb{c},\symb{d}),\symb{f}(\UfC(\symb{c},\symb{d})))} \arU
        &
        \underline{\symb{h}(\symb{f}(\symb{d}),\symb{f}(\symb{f}(\symb{d})))} \arU[dl]
        &
        \symb{h}(\symb{f}(\symb{d}),\symb{f}(\UfC(\symb{c},\symb{d}))) 
        \\
        \guard{\guard{\symb{d}}}
        &
        &
        \symb{d}
        \\
    }
\]
}
\vspace{-20pt}
    \caption{a derivation of $\SR(\cRidea)$ and its corresponding
 computation tree (DAG) of $\U(\cRidea)$.}     
    \label{fig:idea1}
\end{figure}

We capture the observation above by a mapping defined below.
\begin{definition}
 \label{def:Convert} 
 Let $\cR$ be a $\U$-WLL DCTRS. 
 Then, a mapping $\Convert$ from reachable terms (and lists of terms) in
 $T(\ol{\cF},\cV)$ to $2^{T(\U(\cF),\cV)}$ is recursively defined with
 an auxiliary mapping $\Auxconvert$ as follows:
 \begin{itemize}
  \item $\Convert(x) = \{x\}$ for $x \in \cV$,
  \item $\Convert(\olsymb{c}(\MyVec{t}{1}{n})) =
	\{\symb{c}(\MyVec{t'}{1}{n})\mid\MyVec{t'}{1}{n} \in
	\Convert(\MyVec{t}{1}{n}) \}$ for\/ $\symb{c}/n \in \cC_\cR$,
  \item $\Convert(\olsymb{f}(\MyVec{t}{1}{n},\MyVec{u}{1}{m}))
        = \Auxconvert(\olsymb{f}(\MyVec{t}{1}{n},\bot,\ldots,\bot)) \cup
        \bigcup_{i\in\{1,\ldots,m\},u_i\ne\bot}
        \Auxconvert(\olsymb{f}(\MyVec{t}{1}{n},\bot^{i-1},
        u_i,\bot^{m-i})$ for\/ $\symb{f}/n \in \cD_\cR$,
  \item $\Convert(\guard{t}) = \Convert(t)$, 
  \item $\Convert(t)=\emptyset$ where $t$ is not of the form above, 
  \item $\Convert(\epsilon)=\{\epsilon\}$, 
  \item $\Convert(\MyVec{t}{1}{n}) = \{ \MyVec{t'}{1}{n} \mid t'_i \in
	\Convert(t_i), ~ 1 \leq i \leq n \}$ for $n>1$, 
  \item $\Auxconvert(\olsymb{f}(\MyVec{t}{1}{n},\bot,\ldots,\bot)) =
	\{\symb{f}(\MyVec{t'}{1}{n})\mid\MyVec{t'}{1}{n}\in
	\Convert(\MyVec{t}{1}{n}) \}$ for\/ $\symb{f}/n \in
	\cD_\cR$, 
  \item $\Auxconvert(\olsymb{f}(\MyVec{s'}{1}{n},\bot^{i-1}, 
	[\guard{t'},\MyVec{t'}{1}{n_j}]^\rho_j,\bot^{m-i}) =
	\{
	\Usymb^\rho_j(u',\MyVec{u'}{1}{|X_j|})
	\mid
	u' \in \Convert(t'), ~ 
	\MyVec{u'}{1}{|X_j|} \in \Convert(\sigma(\overrightarrow{X_j}))
	\}
	$ for\/ $\symb{f}/n \in \cD_\cR$,
	where
	$\rho:
	\symb{f}(\MyVec{w}{1}{n}) \to r \Leftarrow
	\Condition{s_1}{t_1}, \ldots, \Condition{s_k}{t_k} \in \cR$ is
	the $i$-th conditional rule of\/ $\symb{f}$,
	$\Usymb^{\rho}_{j}(t_j,\overrightarrow{X_j}) \to r'\in
	\U(\rho)$, $X_j =
	\Var(\symb{f}(\MyVec{w}{1}{n}),\MyVec{t}{1}{j-1})$,
	$V_j=\Var(\MyVec{t}{1}{j-1})$, and $\sigma$ is a substitution
	such that $\sigma(\overrightarrow{V_j}) = \MyVec{t'}{1}{n_j}$ 
        and $\sigma\!\left(\MapVec{\Extension}{w}{1}{n}\right)=
	\MyVec{s'}{1}{n}$, 
        and
  \item $\Auxconvert(t)=\emptyset$ where $t$ is not of the form above.
 \end{itemize}
 The mapping $\Convert$ is straightforwardly extended to substitutions
 that have only reachable terms in the range:
 $\Convert(\sigma)=\{\sigma' \mid \Dom(\sigma')\subseteq\Dom(\sigma), \forall x
 \in \Dom(\sigma).~x\sigma' \in \Convert(x\sigma)\}$.
\end{definition}
In applying $\Convert$ to reachable terms, $\Convert$ is never applied
to terms rooted by either $\bot$ or $[~]^\rho_j$.
Though, to simplify proofs below, we define $\Convert$ for $[~]^\rho_j$: 
$\Convert([\,\MyVec{t}{1}{n}\,]^\rho_j)=\MapVec{\Convert}{t}{1}{n}$. 
$\Convert([\,\MyVec{t}{1}{n}\,]^\rho_j)$ is a set of term sequences, and
in the following, we are interested in
$|\Convert([\,\MyVec{t}{1}{n}\,]^\rho_j)|$ rather than elements in
$\Convert([\,\MyVec{t}{1}{n}\,]^\rho_j)$.

We say that a term $s \in T(\ol{\cF},\cV)$ \emph{contains an evaluation
of conditions} if $s$ has a subterm of the form
$\olsymb{f}(\MyVec{t}{1}{n},\ldots,[\ldots]^\rho_j,\ldots)$ for some
$\symb{f}/n \in \cF$. 
We say that a term $\olsymb{f}(\MyVec{t}{1}{n},\MyVec{u}{1}{m}) \in
T(\ol{\cF},\cV)$ \emph{cannot continue any evaluation of conditions at
root position} if for any conditional $\symb{f}$-rule $l \to r
\Leftarrow c \in \cR$, $\Extension(l)$ does not match
$\olsymb{f}(\MyVec{t}{1}{n},\MyVec{u}{1}{m})$. 
We also say that a term $s \in T(\ol{\cF},\cV)$ \emph{cannot continue
any evaluation of conditions} if any subterm of $s$, which is rooted by
a symbol $\olsymb{f}$ with $\symb{f}\in\cD_\cR$, cannot continue any
evaluation of conditions at root position. 
For a term $s$, $|\Convert(s)|=1$ means that $s$ is mapped by $\Convert$
to a unique one in $T(\U(\cF),\cV)$, i.e., $s$ does not contain any
evaluation or cannot continue any evaluation of conditions.
For a substitution $\sigma$ and a term $t$,
$|\Convert(\sigma|_{\Var(t)})| = 1$ means that for each variable $x$ in
$t$, the term substituted for $x$ is mapped by $\Convert$ to a unique
one in $T(\U(\cF),\cV)$, i.e., $|\Convert(x\sigma)|=1$.
The mapping $\Convert$ has the following properties.
\begin{lemma}
 \label{lem:Convert-properties}
 Let $\cR$ be a $\U$-WLL DCTRS, $s$ be a term in $T(\cF,\cV)$, $t$ be a
 reachable term in $T(\ol{\cF},\cV)$, and $\sigma$ be a substitution in
 $\Subst(\ol{\cF},\cV)$. 
 Then, all of the following hold:
 \begin{enumerate}
  \renewcommand{\labelenumi}{(\alph{enumi})}
  \item $\Convert(\ol{s})=\{s\}$.
  \item $\wh{t} \in \Convert(t)$ (i.e., $|\Convert(t)|\geq 1$).
  \item $\{t'\sigma' \mid t' \in \Convert(t),~\sigma' \in
	\Convert(\sigma)\} \subseteq \Convert(t\sigma)$. 
  \item If\/ $\Convert(\sigma|_{\Var(t)})=\{\sigma'\}$ for some
	substitution $\sigma'$ (i.e., $|\Convert(\sigma|_{\Var(t)})| = 1$), 
	then
	$\Convert(\ol{~\wh{t}~}\sigma) = \{ \,\wh{t}\,\sigma'\}$. 
  \item If\/ $\Convert(\sigma|_{\Var(s)})=\{\sigma'\}$ for some
	substitution $\sigma'$ (i.e., $|\Convert(\sigma|_{\Var(s)})| =
	1$), then $\Convert(\ol{s}\sigma) = \{s\sigma'\}$. 
 \end{enumerate}
\end{lemma}
\begin{proofsketch}
 Claims (a) and (b) are trivial by definition. 
 Claim (c) can be proved by structural induction on $t$.
 Claim (d) can be proved analogously to (c) using
 $\Convert(\sigma|_{\Var(t)})=\{\sigma'\}$. 
 Claim (e) is trivial by (d).
\qed
\end{proofsketch}

As illustrated in Figure~\ref{fig:idea1}, our idea is simple and
intuitive. 
Unfortunately, however, the proof for soundness needs some technical
lemmas, while the entire proof is simpler than that
in~\cite{NYG14wpte}. 

Using the mapping $\Convert$ and soundness of $\U$ for WLL DCTRSs, we
show that for a term $s_0 \in T(\cF,\cV)$ and a term $t \in
T(\ol{\cF},\cV)$, if $\guard{\ol{s_0}} \mathrel{\to^*_{\SR(\cR)}} t$,
then $s_0 \mathrel{\to^*_{\U(\cR)}} \wh{t} \in \Convert(t)$
(Lemma~\ref{lem:Un-simulation} and
Theorem~\ref{th:relative-soundness}).
Since $\Convert(\guard{\ol{s_0}})=\{s_0\}$, to show this claim
generally, it suffices to prove the subclaim that for all reachable
terms $s$ and $t$ in $T(\ol{\cF},\cV)$, if $s \mathrel{\to_{l \to r \in
\SR(\cR)}} t$, then for each term $t' \in \Convert(t)$, there exists a
term $s' \in \Convert(s)$ such that $s' \mathrel{\to^*_{\U(\cR)}} t'$.
If $t'$ does not contain a converted term obtained from the reduced
subterm in $t$, then $t'$ is also in $\Convert(s)$.
Otherwise, for the single rewrite step $s \mathrel{\to_{l \to r \in
\SR(\cR)}} t$, one of the following three cases holds:
 \begin{itemize}
  \item The case where $l \to r$ is an auxiliary rule in $\cRaux$. 
	In this case, $\Convert(s) \supseteq \Convert(t)$, and thus, the
	subclaim holds. 
	For example, for any rewrite step by $\cRaux$ in
	Figure~\ref{fig:idea1}, each term in
	$\Convert(t)$ appears in $\Convert(s)$, i.e., for each node $t'$
	for $\Convert(t)$, there exists a node that is for $\Convert(s)$
	and is connected with  $t'$ by the $=$-edge.
    \item The case where $l \to r$ is in $\SR(\cR)\setminus \cRaux$ and
	  $r$ is linear.
	  It is easy to find $s' \in \Convert(s)$ such that $s'$ is
	  reduced by the rule in $\U(\cR)$ corresponding to $l \to r$:
	  $s' \mathrel{\to_{\U(\cR)}} t'$. 
	  In summary, for each $t' \in \Convert(t)$, there exists a term
	  $s'\in \Convert(s)$ such that $s' \mathrel{(= \cup
	  \to_{\U(\cR)})} t'$. 
	  For example, the DAG for $\U(\cRidea)$ in Figure~\ref{fig:idea1}
	  has only $=$- or $\to_{\U(\cRidea)}$-edges because there are
	  only rewrite steps with RL rules in $\SR(\cRidea)$. 
    \item The remaining case where $l \to r$ is in $\SR(\cR) \setminus
	  \cRaux$ and $r$ is \emph{not} linear. 
	  The difficulty of proving the subclaim comes from this case. 
	  We will discuss the detail of the difficulty later.
 \end{itemize}

In proving soundness of $\SR$, neither a variable with non-linear
occurrences nor a non-constructor pattern in the left-hand sides in
$\cR$ is  problematic. 
For example,
$\guard{\olsymb{h}(\olsymb{f}(\guard{\symb{d}},\guard{\symb{c}},\bot),\olsymb{f}(\olsymb{f}(\guard{\symb{d}},\guard{\symb{c}},\bot),\bot,\bot))}$
in Figure~\ref{fig:idea1} is reduced by $\SR(\cRidea)$ to
$\guard{\guard{\symb{d}}}$, but neither
$\symb{h}(\UfC(\symb{c},\symb{d}),\symb{f}(\symb{f}(\symb{d})))$ nor
$\symb{h}(\symb{f}(\symb{d}),\symb{f}(\UfC(\symb{c},\symb{d})))$ in
$\Convert(\guard{\olsymb{h}(\olsymb{f}(\guard{\symb{d}},\guard{\symb{c}},\bot),\olsymb{f}(\olsymb{f}(\guard{\symb{d}},\guard{\symb{c}},\bot),\bot,\bot))})$
can be reduced by $\U(\cRidea)$ to $\symb{d}$.
Though, this is not a problem because for each converted term, we need
the existence of an ancestor but not a descendant.
Viewed in this light, non-left-linearity of rules is not a problem, but
non-right-linearity of rules causes difficulty of proving soundness. 
On the other hand,
$\guard{\olsymb{h}(\symb{d},\olsymb{f}(\symb{d},\bot,\bot))}$ is reduced
by $\SR(\cRidea)$ to
$\olsymb{h}(\symb{d},\olsymb{f}(\symb{d},\guard{\symb{d}},\bot))$, and
then to $\guard{\guard{\symb{d}}}$.
We cannot reduce $\symb{h}(\symb{d},\UfC(\symb{d},\symb{d}))$ in
$\Convert(\guard{\olsymb{h}(\symb{d},\olsymb{f}(\symb{d},\guard{\symb{d}},\bot))})$
by the corresponding rule $\symb{h}(x,\symb{f}(x)) \to \symb{d}$ in
$\U(\cRidea)$ to $\symb{d}$, but another term
$\symb{h}(\symb{d},\symb{f}(\symb{d}))$ in
$\Convert(\guard{\olsymb{h}(\symb{d},\olsymb{f}(\symb{d},\guard{\symb{d}},\bot))})$
can be reduced to $\symb{d}$, simulating the step of $\SR(\cRidea)$.

\begin{figure}[t]
{\scriptsize 
\[
    \xymatrix@R=7pt@C=1.9pt{
        \guard{\olsymb{g}(\underline{\olsymb{f}(\olsymb{a},\bot,\bot)})} \arSR[d]
        &
        &
        &
        \symb{g}(\underline{\symb{f}(\symb{a})}) \ar@{=}[d] \arU[drr]
        &
        &
        \\
        \guard{\olsymb{g}(\olsymb{f}(\olsymb{a},\guard{\underline{\olsymb{a}}},\bot))} \arSR[d]
        &
        &
        &
        \symb{g}(\symb{f}(\underline{\symb{a}})) \ar@{=}[d] 
        &
        &
        \symb{g}(\UfC(\symb{a},\underline{\symb{a}})) \arU[d]
        &
        \\
        \guard{\olsymb{g}(\olsymb{f}(\olsymb{a},\underline{\guard{\guard{\symb{c}}}},\bot))} \arSR[d]
        &
        &
        &
        \symb{g}(\symb{f}(\symb{a})) \ar@{=}[d] 
        &
        &
        \symb{g}(\UfC(\symb{c},\symb{a})) \ar@{=}[d]
        &
        \\
        \guard{\olsymb{g}(\olsymb{f}(\underline{\olsymb{a}},\guard{\symb{c}},\bot))} \arSR[d]
        &
        &
        &
        \symb{g}(\symb{f}(\underline{\symb{a}})) \arUl[dl] \arUs[dd] \arUs[ddr]
        &
        &
        \symb{g}(\UfC(\symb{c},\underline{\symb{a}})) \arU[d] 
        &
        \\
        \guard{\underline{\olsymb{g}(\olsymb{f}(\guard{\symb{d}},\guard{\symb{c}},\bot))}} \arSR[d]
        &
        &
        \underline{\symb{g}(\symb{f}(\symb{d}))} \arU 
        &
        &
        &
        \underline{\symb{g}(\UfC(\symb{c},\symb{d}))} \arU[d] 
        &
        \\
        \guard{\guard{\olsymb{h}(\olsymb{f}(\guard{\symb{d}},\guard{\symb{c}},\bot),\underline{\olsymb{f}(\guard{\symb{d}},\guard{\symb{c}},\bot)})}} \arSR[d]
        &
        &
        \symb{h}(\symb{f}(\symb{d}),\underline{\symb{f}(\symb{d})}) \arUl[dl] \ar@{=}[d]
        &
        \symb{h}(\symb{f}(\symb{d}),\UfC(\symb{c},\symb{d})) \ar@{=}[d]
        &
        \symb{h}(\UfC(\symb{c},\symb{d}),\underline{\symb{f}(\symb{d})}) \arU[dr] \ar@{=}[d]
        &
        \symb{h}(\UfC(\symb{c},\symb{d}),\UfC(\symb{c},\symb{d})) \arU[dr]
        &
        \\
        \guard{\guard{\olsymb{h}(\olsymb{f}(\guard{\symb{d}},\guard{\symb{c}},\bot),\olsymb{f}(\guard{\symb{d}},\guard{\symb{c}},\underline{\guard{\guard{\symb{d}}}}))}} \arSR[d]
        &
        \symb{h}(\symb{f}(\symb{d}),\UfD(\symb{d},\symb{d})) \ar@{=}[d]
        &
        \symb{h}(\symb{f}(\symb{d}),\symb{f}(\symb{d})) \ar@{=}[d] 
        &
        \symb{h}(\symb{f}(\symb{d}),\UfC(\symb{c},\symb{d})) \ar@{=}[d]
        &
        \symb{h}(\UfC(\symb{c},\symb{d}),\symb{f}(\symb{d})) \ar@{=}[d] 
        &
        \symb{h}(\UfC(\symb{c},\symb{d}),\UfD(\symb{d},\symb{d})) \ar@{=}[d]
        &
        \symb{h}(\ldots) \ar@{=}[d] 
        \\
        \guard{\guard{\olsymb{h}(\olsymb{f}(\guard{\symb{d}},\guard{\symb{c}},\bot),\underline{\olsymb{f}(\guard{\symb{d}},\guard{\symb{c}},\guard{\symb{d}})})}} \arSR[d]
        &
        \symb{h}(\symb{f}(\symb{d}),\underline{\UfD(\symb{d},\symb{d})}) \arU[d]
        &
        \symb{h}(\symb{f}(\symb{d}),\symb{f}(\symb{d})) 
        &
        \symb{h}(\symb{f}(\symb{d}),\UfC(\symb{c},\symb{d}))
        &
        \symb{h}(\UfC(\symb{c},\symb{d}),\symb{f}(\symb{d})) 
        &
        \symb{h}(\UfC(\symb{c},\symb{d}),\underline{\UfD(\symb{d},\symb{d})}) \arU
        &
        \symb{h}(\ldots)
        \\
        \guard{\guard{\olsymb{h}(\olsymb{f}(\guard{\symb{d}},\guard{\symb{c}},\bot),\underline{\guard{\symb{d}}})}} \arSR[d]
        &
        \symb{h}(\symb{f}(\symb{d}),\symb{d}) \ar@{=}[d]
        &
        &
        &
        &
        \symb{h}(\UfC(\symb{c},\symb{d}),\symb{d}) \ar@{=}[d]
        &
        \\
        \guard{\guard{\guard{\olsymb{h}(\underline{\olsymb{f}(\guard{\symb{d}},\guard{\symb{c}},\bot)},\symb{d})}}} \arSR[d]
        &
        \symb{h}(\symb{f}(\symb{d}),\symb{d})
        &
        &
        &
        &
        \symb{h}(\underline{\UfC(\symb{c},\symb{d})},\symb{d}) \arU
        &
        \\
        \guard{\guard{\guard{\olsymb{h}(\underline{\guard{\symb{c}}},\symb{d})}}} \arSR[d]
        &
        &
        &
        &
        &
        \symb{h}(\symb{c},\symb{d}) \ar@{=}[d]
        &
        \\
        \guard{\guard{\guard{\guard{\underline{\olsymb{h}(\symb{c},\symb{d})}}}}} \arSR[d]
        &
        &
        &
        &
        &
        \underline{\symb{h}(\symb{c},\symb{d})} \arU
        &
        \\
        \guard{\guard{\guard{\guard{\guard{\symb{c}}}}}}
        &
        &
        &
        &
        &
        \symb{c}
        &
        \\
    }
\]
}
\vspace{-20pt}
    \caption{another derivation of $\SR(\cRidea)$ and its corresponding computation tree (DAG) of $\U(\cRidea)$.}    
    \label{fig:idea2}
\end{figure}
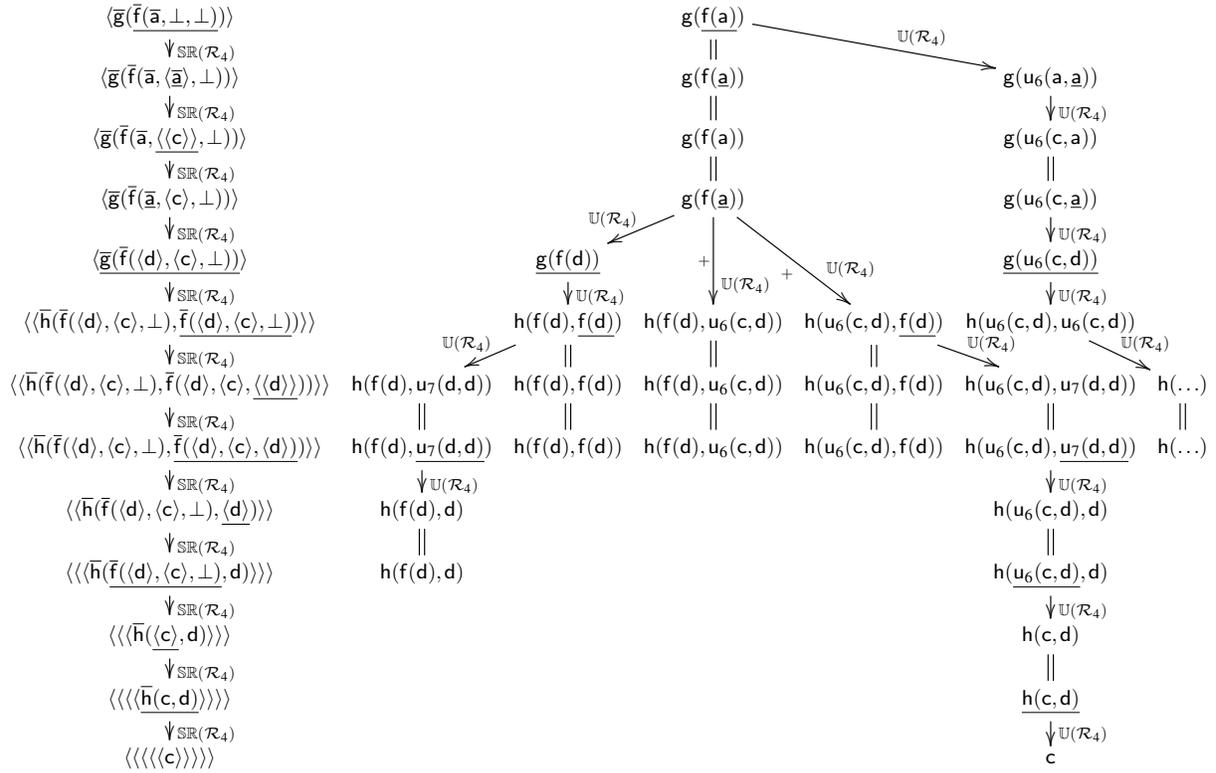

Let us get back to the case where we apply non-right-linear rules in
$\SR(\cR)\setminus \cRaux$. 
Figure~\ref{fig:idea2} illustrates a derivation of $\SR(\cRidea)$ and
its corresponding computation tree of $\U(\cRidea)$, where
non-right-linear rules are applied. 
In applying non-right-linear rules in $\SR(\cR)\setminus \cRaux$ to $s$
with $s \mathrel{\to_{l \to r \in \SR(\cR)\setminus \cRaux}} t$, it is
not only difficult but also sometimes impossible to show that for each
$t' \in \Convert(t)$, there exists a term $s' \in \Convert(s)$ such that
$s' \mathrel{\to^*_{\U(\cR)}} t'$. 
For example,
$\guard{\olsymb{g}(\olsymb{f}(\guard{\symb{d}},\guard{\symb{c}},\bot))}$
in Figure~\ref{fig:idea2} is reduced by $\SR(\cRidea)$ to
$\guard{\guard{\olsymb{h}(\olsymb{f}(\guard{\symb{d}},\guard{\symb{c}},\bot),\olsymb{f}(\guard{\symb{d}},\guard{\symb{c}},\bot))}}$,
and represents two terms in $T(\U(\cF),\cV)$:
$\Convert(\guard{\olsymb{g}(\olsymb{f}(\guard{\symb{d}},\guard{\symb{c}},\bot))})
= \{~ 
\symb{g}(\symb{f}(\symb{d})),~
\symb{g}(\UfC(\symb{c},\symb{d}))
~\}$.
Though, no term in
$\Convert(\guard{\olsymb{g}(\olsymb{f}(\guard{\symb{d}},\guard{\symb{c}},\bot))})$
is reduced by a single step of $\to_{\SR(\cRidea)}$ to either
$\symb{h}(\symb{f}(\symb{d}),\UfC(\symb{c},\symb{d}))$ or
$\symb{h}(\UfC(\symb{c},\symb{d}),\symb{f}(\symb{d}))$ in
$\Convert(\guard{\guard{\olsymb{h}(\olsymb{f}(\guard{\symb{d}},\guard{\symb{c}},\bot),\olsymb{f}(\guard{\symb{d}},\guard{\symb{c}},\bot))}})$. 
However, $\symb{g}(\symb{f}(\symb{a})) \in
\Convert(\guard{\olsymb{g}(\olsymb{f}(\symb{a},\guard{\symb{c}},\bot))})$
can be reduced to both
$\symb{h}(\symb{f}(\symb{d}),\UfC(\symb{c},\symb{d}))$ and
$\symb{h}(\UfC(\symb{c},\symb{d}),\symb{f}(\symb{d}))$ including rewrite
steps of $\symb{g}(x) \to \symb{h}(x,x) \in \U(\cRidea)$ corresponding
to $\olsymb{g}(x) \to \guard{\olsymb{h}(x,x)} \in \SR(\cRidea)$.
The existence of such reduction sequences represented by
$\to^+_{\U(\cRidea)}$-edges of the DAG in Figure~\ref{fig:idea2} is
ensured by  the reduction
$\guard{\olsymb{g}(\olsymb{f}(\olsymb{a},\guard{\symb{c}},\bot))}
\mathrel{\to^*_{\SR(\cRidea)}}
\guard{\guard{\olsymb{h}(\olsymb{f}(\guard{\symb{d}},\guard{\symb{c}},\bot),\olsymb{f}(\guard{\symb{d}},\guard{\symb{c}},\bot))}}$. 
In summary, for $s_0 \mathrel{\to^*_{\SR(\cR)}} s \mathrel{\to_{l \to r
\in \SR(\cR)\setminus \cRaux}} t$, we will show the existence of a
derivation $s_0 \mathrel{\to^*_{\SR(\cR)}} s' \mathrel{\to_{l \to r \in
\SR(\cR)\setminus\cRaux}} t' \mathrel{\to^*_{\SR(\cR)}} t$ such that for
each term $t'' \in \Convert(t')$, $s'' \mathrel{\to_{l' \to r' \in
\U(\cR)}} t''$ for some term $s'' \in \Convert(s')$, where $l' \to r'$
corresponds to $l \to r$ (Lemmas \ref{lemma:SR-properties-2}
and~\ref{lemma:SR-properties-3}).

Before showing the key lemmas (Lemmas \ref{lemma:SR-properties-2}
and~\ref{lemma:SR-properties-3}), we show some auxiliary lemmas along
the intuition above. 
The following lemma says that if a term $t$ is reduced and matches a
linear pattern obtained from $\cR$ by applying $\Extension$, then the 
initial term $t$ also matches the pattern. 
\begin{lemma}
 \label{lemma:SR-properties}
 Let $\cR$ be a $\U$-WLL DCTRS, $t$ be a reachable term in
 $T(\ol{\cF},\cV)$, $w$ be a linear term in $T(\cF,\cV)$,
 $w'=\Extension(w)$ (i.e., $w' \in T(\ol{\cF}\setminus \{ \guard{},
 \bot, [~]^\rho_j \},\cV)$),%
 \footnote{ Patterns
 $w'_1,\ldots,w'_n,\Extension(t_1),\ldots,\Extension(t_k)$ in
 Definition~\ref{def:SR} are in
 $T(\ol{\cF}/\{\guard{},\bot,[~]^\rho_j\},\cV)$.}
 and
 $\theta$ be a substitution in $\Subst(\ol{\cF},\cV)$. 
 If $t \mathrel{\to^*_{\SR(\cR)}} w'\theta$, then there exists a
 substitution $\sigma$ such that $t=w'\sigma$ and $x\sigma
 \mathrel{\to^*_{\SR(\cR)}} x\theta$ for all variables $x\in\Var(w)$.
\end{lemma}
\begin{proofsketch}
 It suffices to show that if $t \mathrel{\to_{p,l \to r \in \SR(\cR)}}
 w'\theta$, then there exists a substitution $\sigma$ such that
 $t=w'\sigma$ and $x\sigma \mathrel{\to^*_{\SR(\cR)}} x\theta$ for all
 variables $x\in\Var(w)$.
 This claim can be proved by structural induction on $w$.
\qed
\end{proofsketch}

For the sake of readability, we introduce a binary relation $\phito$
over $T(\ol{\cF},\cV)$:
$s \mathrel{\phito} t$ if and only if for each $t' \in \Convert(t)$,
there exists a term $s'\in \Convert(s)$ such that $s'
\mathrel{\to^*_{\U(\cR)}} t'$.
It is clear that $\phito$ is reflexive and transitive. 
The relation $\phito$ is closed under contexts.
\begin{lemma}
 \label{lem:phito_closed-under-contexts}
 Let $\cR$ be a $\U$-WLL DCTRS, $C[~]$ be a context, and $s$ and $t$ be
 terms in $T(\ol{\cF},\cV)$ such that $s \mathrel{\to^*_{\SR(\cR)}} t$
 and $s \mathrel{\phito} t$. 
 Then, $C[s] \mathrel{\phito} C[t]$.
\end{lemma}
\begin{proofsketch}
Using the definition of $\Convert$ and Lemma~\ref{lemma:SR-properties},
 this lemma can be proved by structural induction on $C[~]$.
\qed
\end{proofsketch}
The following lemma is a variant of
Lemma~\ref{lem:phito_closed-under-contexts} that is useful to prove the
main key lemma shown later.
\begin{lemma}
 \label{lem:phito_closed-under-substitutions}
 Let $\cR$ be a $\U$-WLL DCTRS, $t$ be a term in $T(\ol{\cF},\cV)$, and
 $\sigma$ and $\theta$ be substitutions such that $x\sigma
 \mathrel{\to^*_{\SR(\cR)}} x\theta$ and $x\sigma \mathrel{\phito}
 x\theta$ for all variables $x \in \Var(t)$. 
 Then, $t\sigma \mathrel{\phito} t\theta$.
\end{lemma}
\begin{proofsketch}
 It is easy to extend Lemma~\ref{lem:phito_closed-under-contexts} to
 contexts with multiple holes. 
 Thus, this lemma is a direct consequence of the extended lemma since a
 linear term can be considered a context with multiple holes.
\qed    
\end{proofsketch}

When $l \to r \in \SR(\cR)$ has a variable $x$ such that $|r|_x > 1$ and
$|\Convert(x\theta)|>1$, we have at least two terms obtained by
converting $x\theta$, and thus, $\Convert(r\sigma)$ contains a term that
has no ancestor in $\Convert(l\theta)$ w.r.t.\ $\to_{\U(\cR)}$.
This problem does not happen if $\Convert(\theta|_{\Var(r)})$ is a
singleton set.
\begin{lemma}
 \label{lemma:SR-properties-2}
 Let $\cR$ be a $\U$-WLL DCTRS, $l \to r \in \SR(\cR)$, and $\sigma$ be
 a substitution such that for any variable $x \in \Var(r)$, if\/ $|r|_x
 > 1$ and $x\sigma\ne \bot$ then $|\Convert(x\sigma)|=1$.
 Then, $l\sigma \mathrel{\phito} r\sigma$.
\end{lemma}
\begin{proofsketch}
 Referring to the definition of $\Convert$ and
 Lemma~\ref{lem:Convert-properties}, this lemma can be proved by a case
 distinction depending on what $l \to r$ is.
 \qed
\end{proofsketch}

For a derivation $s \mathrel{\to^*_{\SR(\cR)}} t\theta$, the following
lemma ensures the existence of an ancestor for a variable $x$ in $t$
such that $|t|_x = 1$ and $|\Convert(x\theta)|>1$.
\begin{lemma}
 \label{lemma:SR-properties-3}
 Let $\cR$ be a $\U$-WLL DCTRS, $s$ and $t$ be terms in
 $T(\ol{\cF},\cV)$, $\theta$ be a substitution in
 $\Subst(\ol{\cF},\cV)$, and $X\subseteq\{ x \in \Var(t) \mid |t|_x =1
 \}$. 
 If\/ $|\Convert(s)|=1$ and $s \mathrel{\to_{\SR(\cR)}^d} t\theta$ ($d
 \geq 0$), then there exist a substitution $\delta \in
 \Subst(\ol{\cF},\cV)$ and natural numbers $d'$ and $d_x$ for $x\in X$
 such that 
 \begin{enumerate}
  \renewcommand{\labelenumi}{(\alph{enumi})}
  \item $d'+\sum_{x \in X} d_x \leq d$, 
  \item $s \mathrel{\to^{d'}_{\SR(\cR)}} t\delta$, 
  \item $|\Convert(x\delta)|=1$ and $x\delta
	\mathrel{\to^{d_x}_{\SR(\cR)}} x\theta$ for all variables $x \in
	X$ such that $x\theta\ne \bot$, and 
  \item $x\delta=x\theta$ for all variables $x \in \Var(t) \setminus \{
	y \in X \mid y\theta\ne \bot \}$.%
	\footnote{ Note that if $x\theta=\bot$, then $x\delta=\bot$.}
 \end{enumerate}
\end{lemma}
\begin{proofsketch}
Using Theorem~\ref{thm:UWLL-properties},
 Lemma~\ref{lem:Convert-properties}~(a),~(c),~(d), and~(e), and
 Lemma~\ref{lemma:SR-properties}, this lemma can be proved by induction
 on the lexicographic product of $d$ and the size of $s$.
 \qed
\end{proofsketch}
We show the main key lemma on the relationship between
$\to^*_{\SR(\cR)}$ and $\phito$.
\begin{lemma}
 \label{lem:Un-simulation}
 Let $\cR$ be a $\U$-WLL DCTRS, and 
  $s$ and $t$ be terms in $T(\ol{\cF},\cV)$. 
 If\/ $|\Convert(s)|=1$ and $s \mathrel{\to^d_{\SR(\cR)}} t$ ($d\geq
 0$), then $s \mathrel{\phito} t$.
\end{lemma}
\begin{proofsketch}
 Using Theorem~\ref{thm:UWLL-properties},
 Lemmas~\ref{lem:phito_closed-under-contexts},~\ref{lemma:SR-properties-2}
 and~\ref{lemma:SR-properties-3}, this lemma can be proved by induction
 on $d$. 
\qed
\end{proofsketch}

Finally, we show the key result of this paper. 
\begin{theorem} 
 \label{th:relative-soundness}
 $\SR$ is sound for WLL and\/ $\U$-WLL DCTRSs.
\end{theorem}
 \begin{proof}
  Let $\cR$ be a WLL and $\U$-WLL DCTRS over a signature $\cF$, $s \in
  T(\cF,\cV)$, and $t \in T(\ol{\cF},\cV)$. 
  Suppose that $\guard{\ol{s}} \mathrel{\to^*_{\SR(\cR)}} t$.  
  It follows from Lemma~\ref{lem:Un-simulation} that $\guard{\ol{s}}
  \mathrel{\phito} t$. 
  It follows from Lemma~\ref{lem:Convert-properties}~(a),~(b) that
  $\Convert(\guard{\ol{s}})=\Convert(\ol{s})=\{s\}$ and $\wh{t} \in
  \Convert(t)$, and hence, by the definition of $\phito$, $s
  \mathrel{\to^*_{\U(\cR)}} \wh{t}$. 
  Since $\U$ is sound for $\cR$ by Theorem~\ref{th:U-soundness_for_WLL},
  it holds that $s \mathrel{\to^*_\cR} \wh{t}$. 
  Therefore, $\SR$ is sound for $\cR$. 
\qed
 \end{proof}

Let us consider the conversion $\T$ in Section~\ref{def:U} again.
As a consequence of Theorems~\ref{thm:from_WLL_to_UWLL}
and~\ref{th:relative-soundness}, we show that the composed
transformation $\SR\circ\T$ of a WLL DCTRS into a WLL TRS is
sound for WLL DCTRSs.
 \begin{theorem}
  \label{thm:correctness_of_T}
  The composed transformation $\SR\circ\T$ is sound and complete for WLL
  DCTRSs.
 \end{theorem}

\section{Conclusion}
\label{sec:conclusion}

In this paper, we have shown that every WLL DCTRS can be converted to an
equivalent WLL and $\U$-WLL DCTRS and the {\SRtransformation} is
applicable to $\U$-WLL DCTRSs without any change.
Then, we have proved that the {\SRtransformation} is sound for WLL and
$\U$-WLL DCTRSs.
As a consequence of these results, we have shown that the composition of
the conversion and the {\SRtransformation} is a sound
structure-preserving transformation for WLL DCTRSs.
For computational equivalence to WLL SDCTRSs, we have to show that if
$\cR$ is confluent, then so is $\SR(\cR)$. To prove this claim as
in~\cite{SR06b} is one of our future works.
To expand the applicability of the {\SRtransformation}, we will extend
the {\SRtransformation} to other classes.  

\paragraph{Acknowledgements}

We thank the anonymous reviewers very much for their useful comments to
improve this paper.





\end{document}